\documentclass[11pt]{article}
\usepackage[margin=1in]{geometry}                
\usepackage{times}

\usepackage{graphicx}
\usepackage{authblk}
\usepackage{amssymb,amsmath, amsthm}
\usepackage{epstopdf}
\usepackage{url}
\usepackage{pbox}
\usepackage{caption}
\usepackage{subcaption}

\usepackage{color}

\title{Verifying {Chemical Reaction Network} Implementations: \\A
  Pathway Decomposition Approach \thanks{Preliminary versions of
    this manuscript appeared in the proceedings of VEMDP 2014 {and are available on arXiv:1411.0782 [cs.CE]}.}}
\author[$1$]{Seung Woo Shin}
\author[$2$]{Chris Thachuk}
\author[$2$]{Erik Winfree}
\affil[$1$]{\textit{University of California, Berkeley}}
\affil[$2$]{\textit{California Institute of Technology}}
\date{}                                           

\usepackage{amsthm,amssymb}

\theoremstyle{definition}
\newtheorem{definition}{Definition}
\newtheorem{theorem}{Theorem}[section]
\newtheorem{lemma}[theorem]{Lemma}
\newtheorem{corollary}[theorem]{Corollary}

\usepackage{accsupp}
\newcommand*{\llbrace}{%
  \BeginAccSupp{method=hex,unicode,ActualText=2983}%
    \textnormal{\usefont{OMS}{lmr}{m}{n}\char102}%
    \mathchoice{\mkern-4.05mu}{\mkern-4.05mu}{\mkern-4.3mu}{\mkern-4.8mu}%
    \textnormal{\usefont{OMS}{lmr}{m}{n}\char106}%
  \EndAccSupp{}%
}
\newcommand*{\rrbrace}{%
  \BeginAccSupp{method=hex,unicode,ActualText=2984}%
    \textnormal{\usefont{OMS}{lmr}{m}{n}\char106}%
    \mathchoice{\mkern-4.05mu}{\mkern-4.05mu}{\mkern-4.3mu}{\mkern-4.8mu}%
    \textnormal{\usefont{OMS}{lmr}{m}{n}\char103}%
  \EndAccSupp{}%
}

\begin{document}
\maketitle
{\abstract
The emerging fields of genetic engineering, synthetic biology, DNA computing, DNA nanotechnology, and molecular programming herald the birth of a new information technology that acquires information by directly sensing molecules within a chemical environment, stores information in molecules such as DNA, RNA, and proteins, processes that information by means of chemical and biochemical transformations, and uses that information to direct the manipulation of matter at the nanometer scale.   To scale up beyond current proof-of-principle demonstrations, new methods for managing the complexity of designed molecular systems will need to be developed.  Here we focus on the challenge of verifying the correctness of molecular implementations of abstract chemical reaction networks, where operation in a well-mixed ``soup'' of molecules is stochastic, asynchronous, concurrent, and often involves multiple intermediate steps in the implementation, parallel pathways, and side reactions.  
{This problem relates to the verification of Petri nets, but existing approaches are not sufficient for providing a single guarantee covering an infinite set of possible initial states (molecule counts) and an infinite state space potentially explored by the system given any initial state.   We address these issues by formulating a new theory of pathway decomposition that provides an elegant formal basis for comparing chemical reaction network implementations, and we present an algorithm that computes this basis.  Our theory naturally handles certain situations that commonly arise in molecular implementations, such as what we call ``delayed choice,'' that are not easily accommodated by other approaches.}
We further show how pathway decomposition can be combined with weak bisimulation to handle a wider class that includes {most} currently known enzyme-free DNA implementation techniques.  We anticipate that our notion of logical equivalence between chemical reaction network implementations will be valuable for other molecular implementations such as biochemical enzyme systems, and perhaps even more broadly in concurrency theory.
}


\section{Introduction} \label{section:introduction}

A central problem in molecular computing and bioengineering is that of implementing algorithmic behavior using chemical molecules. The ability to design chemical systems that can sense and react to the environment finds applications in many different fields, such as nanotechnology \cite{chen2013programmable}, medicine \cite{douglas2012logic}, and robotics \cite{gu2010proximity}. Unfortunately, the complexity of such engineered chemical systems often makes it challenging to ensure that a designed system really behaves according to specification. 
Furthermore, since experimentally synthesizing chemical systems can require considerable resources, mistakes are generally expensive, and it would be useful to have a procedure by which one can theoretically verify the correctness of a design using computer algorithms prior to synthesis.
In this paper we propose a theory that can serve as a foundation for such automated verification procedures.

Specifically, we focus our attention on the problem of verifying chemical reaction network (CRN) implementations. Informally, a CRN is a set of chemical reactions that specify the behavior of a given chemical system in a well mixed solution.  For example, the reaction equation $A + B \rightarrow C$ means that a reactant molecule of type $A$ and another of type $B$ can be \textit{consumed} in order to \textit{produce} a product molecule of type $C$.  A reaction is applicable if all of its reactants are present in the solution in sufficient quantities. In case both $A+B\to C$ and $C\to A+B$ are in the CRN, we may also use the shorthand notation $A+B\leftrightharpoons C$. In general, the evolution of the system from some initial set of molecules is a stochastic, asynchronous, and concurrent process.
 While abstract CRNs provide the most widely used formal language for describing chemical systems, and have done so for over a century, only recently have abstract CRNs been used explicitly as a programming language in molecular programming and bioengineering. This is because CRNs are often used to specify the target behavior for an engineered chemical system (see Figure \ref{fig:crn_implementation}).  How can one realize these ``target'' CRNs experimentally?  Unfortunately, synthesizing chemicals to efficiently interact --- and only as prescribed --- presents a significant, if not infeasible, engineering challenge.  Fortunately, any target CRN can be \textit{emulated} by a (generally more complex) ``implementation'' CRN.  For example, in the field of DNA computing, implementing a given CRN using synthesized DNA strands is a well studied topic that has resulted in a number of translation schemes \cite{SSW10, C09, QSW11}.

In order to evaluate CRN implementations prior to their experimental
demonstration, a mathematical model describing the expected molecular
interactions is necessary.  For this purpose, software simulators that
embody the relevant physics and chemistry can be used.  Beyond
performing simulations -- which by themselves can't provide absolute
statements about the correctness of an implementation -- it is often
possible to describe the model of the molecular implementation as a
CRN.  That is, software called ``reaction enumerators'' can, given a set of initial molecules, evaluate all possible configuration changes and interactions, possibly generating new molecular species, and repeating until the full set of species and reactions have been enumerated. In the case of DNA systems, there are multiple software packages available for this task~\cite{DSD, Grun}.   {More general biochemical implementations could be modeled using languages such as BioNetGen~\cite{BNGL} and Kappa~\cite{danos2007}.}

Given a ``target'' CRN which specifies a desired algorithmic behavior and an ``implementation'' CRN which purports to implement the target CRN, how can one check that the implementation CRN is indeed correct? As we shall see, this question involves subtle issues that make it difficult to even define a notion of correctness that can be universally agreed upon, despite the fact that in this paper we study a somewhat simpler version of the problem in which chemical kinetics, i.e.~rates of chemical reactions, is dropped from consideration. However, we note that this restriction is not without its own advantages. For instance, when basing a theory on chemical kinetics, it is of interest to accept approximate matches to the target behavioral dynamics \cite{Klavins1, Klavins2}, which may overlook certain logical flaws in the implementation that occur rarely.  While theories of kinetic equivalence are possible and can in principle provide guarantees about timing \cite{C14}, they can be difficult to apply to molecular engineering in practice.  In contrast, a theory that ignores chemical kinetics can be exact and therefore emphasize the logical aspect of the correctness question.

The main challenge in this verification problem lies in the fact that the implementation CRN is usually much more complex than the target CRN. This is because each reaction in the target CRN, which is of course a single step in principle, gets implemented as a sequence of steps which may involve ``intermediate'' species that were not part of the original target CRN. For example, in DNA-based implementations, the implementation CRN can easily involve an order of magnitude more reactions and species than the target CRN (the size will depend upon the level of detail in the model of the implementation \cite{DSD, Grun, Schaeffer, DOL13}). Given that the intermediate species participating in implementations of different target reactions can potentially interact with each other in spurious ways, it becomes very difficult to verify that such an implementation CRN is indeed ``correct.''

\begin{figure}[!h]
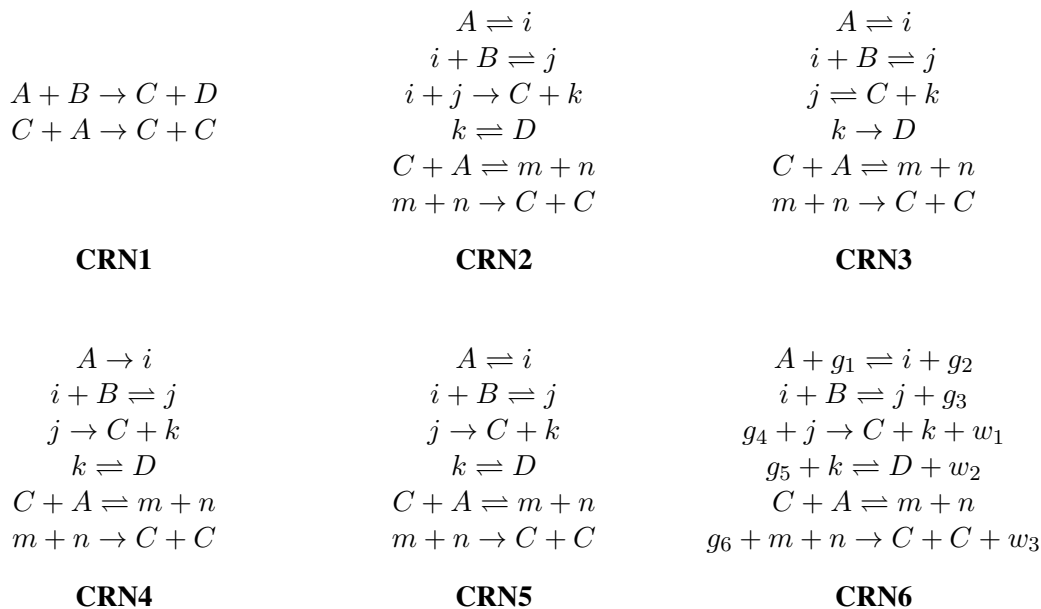

\centering
\vspace{0.7cm}
\begin{minipage}[b]{0.3\linewidth}
\centering
$A+B\to C+D$\\
$C+A\to C+C$\\
\ \\
\ \\
\vspace{0.3cm}
\textbf{CRN1}
\end{minipage}
\begin{minipage}[b]{0.3\linewidth}
\centering
$A\rightleftharpoons i$\\
$i+B\rightleftharpoons j$\\
$i+j\to C+k$\\
$k\rightleftharpoons D$\\
$C+A\rightleftharpoons m+n$\\
$m+n\to C+C$\\
\vspace{0.3cm}
\textbf{CRN2}
\end{minipage}
\begin{minipage}[b]{0.3\linewidth}
\centering
$A\rightleftharpoons i$\\
$i+B\rightleftharpoons j$\\
$j\rightleftharpoons C+k$\\
$k\to D$\\
$C+A\rightleftharpoons m+n$\\
$m+n\to C+C$\\
\vspace{0.3cm}
\textbf{CRN3}
\end{minipage}
\begin{minipage}[b]{0.3\linewidth}
\vspace{1cm}
\centering
$A\to i$\\
$i+B\rightleftharpoons j$\\
$j\to C+k$\\
$k\rightleftharpoons D$\\
$C+A\rightleftharpoons m+n$\\
$m+n\to C+C$\\
\vspace{0.3cm}
\textbf{CRN4}
\end{minipage}
\begin{minipage}[b]{0.3\linewidth}
\vspace{1cm}
\centering
$A\rightleftharpoons i$\\
$i+B\rightleftharpoons j$\\
$j\to C+k$\\
$k\rightleftharpoons D$\\
$C+A\rightleftharpoons m+n$\\
$m+n\to C+C$\\
\vspace{0.3cm}
\textbf{CRN5}
\end{minipage}
\begin{minipage}[b]{0.3\linewidth}
\vspace{1cm}
\centering
$A+g_1\rightleftharpoons i+g_2$\\
$i+B\rightleftharpoons j+g_3$\\
$g_4+j\to C+k+w_1$\\
$g_5+k\rightleftharpoons D+w_2$\\
$C+A\rightleftharpoons m+n$\\
$g_6+m+n\to C+C+w_3$\\
\vspace{0.3cm}
\textbf{CRN6}
\end{minipage}
\caption{An example of CRN implementation. CRN1 represents the ``target'' CRN, i.e.,~the behavior we desire to implement, whereas CRN2-5 are potential ``implementations'' of this target CRN. In these CRNs, the lowercase species are ``intermediate'' species of the implementations, while the uppercase species are ``formal'' species. CRN6 illustrates the way in which ``fuel'' and ``waste'' species may appear in a typical DNA-based system, with fuel species denoted by $g_i$ and waste species denoted by $w_i$.  Removing inert waste species and ever-present fuel species from CRN6 yields CRN5.}\label{fig:crn_implementation}\end{figure}

It is not immediately obvious how to precisely define what makes an implementation correct or incorrect, so it is helpful to informally examine a few examples. Figure \ref{fig:crn_implementation} illustrates various different ways that a proposed implementation can be ``incorrect.'' For instance, one can easily see that CRN2 is clearly not a good implementation of CRN1, because it implements the reaction $A+A+B\to C+D$ in place of $A+B\to C+D$. CRN3 is incorrect in a more subtle way. While a cursory look may not reveal any immediate problem with this implementation, one can check that CRN3 can get from the initial state\footnote{In this paper, we use the notation $\llbrace \cdot \rrbrace$ to denote multisets.}  $\llbrace A,A,B\rrbrace$ to a final state $\llbrace A,B,C\rrbrace$, whereas there is no way to achieve this using reactions from CRN1.\footnote{The pathway is 
($A\to i$, $i+B\to j$, $j\to C +k$, $C+A \to m+n$, $m+n\to C+C$, $C+k\to j$, $j\to i+B$, $i\to A$).
}
CRN4 is incorrect in yet another way. Starting from the initial state $\llbrace A,C\rrbrace$, one can see that the system will sometimes get ``stuck'' in the state $\llbrace i, C\rrbrace$,
unable to produce $\llbrace C, C\rrbrace$,
with $i$ becoming an intermediate species that is not really ``intermediate.'' Now, CRN5 seems to be free of any such issue, but with what confidence can we declare that it is a correct implementation of CRN1, having seen the subtle ways that an implementation can go wrong? A goal of this paper is to provide a mathematical definition of ``correctness'' of CRN implementations which can be used to test them in practice.

In our further discussions, we will restrict our attention to
implementation CRNs that satisfy the condition that we call
``tidiness.'' Informally stated, tidy CRNs are implementation CRNs
which do not get ``stuck'' in the way that CRN4 got stuck above,  i.e., they always can ``clean up'' intermediate species. This means that any intermediate species that are produced during the evolution of the system can eventually turn back into species of the target CRN. Of course, the algorithm we present in this paper for testing our definition of correctness will also be able to test whether the given implementation is tidy.

Finally, we briefly mention that many CRN implementations also involve what are called ``fuel'' and ``waste'' species, in addition to the already mentioned intermediate species. Fuel species are helper species that are assumed to be always present in the system at fixed concentration, whereas waste species are chemically inert species that sometimes get produced as a byproduct of implemented pathways (see CRN6 of Figure \ref{fig:crn_implementation} or for a more detailed explanation Example \#1 of Section \ref{sec:case}). While our core theory addresses the version of the problem in which there is no fuel or waste species, as we demonstrate in Section \ref{section:discussion}, it can easily be extended to handle the general case with fuel and waste species, using existing tools.

\section{Motivations for a new theory}\label{section:motivation}

To one who is experienced in formal verification, the problem seems to be closely related to various well-studied notions such as reachability, (weak) trace equivalence, (weak) bisimulation, serializability, etc. In this section, we briefly demonstrate why none of these traditional notions seems to give rise to a definition which is entirely satisfactory for the problem at hand.

The first notion we consider is reachability between formal
states~~\cite{Mayr1981,Lipton1976,EN94}.  We call the species that
appear in both the target and the implementation CRNs ``formal,'' to
distinguish them from species that appear only in the implementation
CRN, which we call ``intermediate.'' Formal states are defined to be
states which do not contain any intermediate species. Since we are
assuming that our implementation CRN is tidy, it then makes sense to
ask whether the target CRN and the implementation CRN have the same
reachability when we restrict our attention to formal states only ---
this is an important distinction from the traditional Petri net
reachability-equivalence problem. That is, given some formal state,
what is the set of formal states that can be reached from that state
using reactions from one CRN, as opposed to the other CRN? Do the
target CRN and the implementation CRN give rise to exactly the same
reachability for every formal initial state? While it is obvious that
any ``correct'' implementation must satisfy this condition, it is also
easy to see that this notion is not very strong. For example, consider
the target CRN $\{A\to B,\ B\to C,\ C\to A\}$ and the implementation
CRN $\{A\to i,\ i\to C,\ C\to j,\ j\to B,\ B\to k,\ k\to A\}$. The two CRNs are implementing opposite behaviors in the sense that
starting from one $A$ molecule, the target CRN will visit formal states in the clockwise
order $\llbrace A\rrbrace, \llbrace B\rrbrace, \llbrace C\rrbrace, \llbrace A\rrbrace, \llbrace B\rrbrace, \llbrace C\rrbrace, \ldots$, whereas the implementation CRN will visit formal states in the counter-clockwise order $\llbrace A\rrbrace, \llbrace C\rrbrace, \llbrace B\rrbrace, \llbrace A\rrbrace, \llbrace C\rrbrace, \llbrace B\rrbrace, \ldots$. Nonetheless, they still give rise to the same reachability between purely formal states.

Trace equivalence \cite{EN94,PEM99} is another notion of equivalence that is often found in formal verification literature.
To our knowledge, it has not been applied in the context of CRN equivalence.
We interpret its application in this context as follows. 
\textit{Weak} trace equivalence requires that it should be possible to ``label'' the reactions of the implementation CRN to be either a reaction of the target CRN or a ``null'' reaction.
This labeling must be such that for any formal initial state, any sequence of reactions that can take place in the target CRN should also be able to take place in the implementation CRN and vice versa, up to the interpretation specified by the given labeling.
However, it turns out to be an inappropriate notion in our setting.
For example, consider the target CRN $\{A\rightleftharpoons B, B\rightleftharpoons C, C\rightleftharpoons A\}$ and the implementation CRN $\{A\rightleftharpoons i, B\rightleftharpoons i, C\rightleftharpoons i\}$.
The dynamics of the implementation appear correct since each reaction of the target CRN can be simulated in the implementation CRN in the obvious way by exactly two reactions:
the first reaction consumes the reactant and produces an intermediate species $i$ while the second reaction consumes $i$ and produces the intended formal species.
However, these CRNs are \emph{not} (weak-)trace equivalent.
Consider that every reaction of the implementation CRN must be labeled by one of the six formal reactions (since the implementation CRN also consists of six reactions) and none can be labeled as a ``null'' reaction.
Since any initial reaction of the implementation CRN must begin in a formal state, and since there are only three reactions that can occur from one of the three formal states, then any trace of the target CRN that begins with one of the \emph{other} three possible reactions cannot be simulated by the implementation CRN.
Consider a second example with target CRN $\{A\rightleftharpoons B, B\to C, C\to A\}$ and implementation CRN $\{A\rightleftharpoons B, B\to j, j\to C, C\to A, C\to \emptyset, \emptyset\rightleftharpoons i\}$ where the implementation reactions $\{j\to C, C\to \emptyset, \emptyset\to i, i\to \emptyset\}$ are labeled as ``null'' and the other reactions are labeled in the obvious way that is consistent with formal species names.
The implementation CRN exemplifies a common shortcoming of trace equivalence: inability to distinguish the two systems with respect to \textit{deadlock}.
In our example the implementation CRN \textit{can in principle} simulate all finite and infinite traces of the target CRN, but once the first ``null'' reaction $C\to \emptyset$ occurs then only ``null'' reactions can follow.
In essence, the implementation CRN can become ``stuck'' whereas the target CRN cannot.
While (weak-)trace equivalence cannot distinguish based on deadlock conditions as in our second example, other equivalence notions such as bisimulation can.

Bisimulation{~\cite{milner1989,buchholz2008}}  is perhaps the most influential notion of equivalence in state transition systems such as CRNs, Petri nets, or concurrent systems{~\cite{Heiner2008,SS00,DK05}}.
A notion of CRN equivalence based on the idea of weak bisimulation is explored in detail in \cite{Qing,Johnson2016}, and indeed it proves to be much more useful than the above two notions. For bisimulation equivalence of CRNs, each intermediate species is ``interpreted'' as some combination of formal species, such that in any state of the implementation CRN, the set of possible next non-trivial reactions is exactly the same as it would be in the formal CRN. 
(Here, a ``trivial'' reaction is one where the interpretation of the reactants is identical to the interpretation of the products.)
However, one potential problem of this approach is that it demands a way of interpreting every intermediate species in terms of formal species. Therefore, if we implement the target CRN $\{A\to B,\ A\to C,\ A\to D\}$ as $\{A\to i,\ i\to B,\ i\to C,\ A\to j,\ j\to D\}$, we cannot apply this bisimulation approach because the intermediate $i$ cannot be interpreted to be any of $A$, $B$, or $C$. Namely, calling it $A$ would be a bad interpretation because $i$ can never turn into $D$. Calling it $B$ would be bad because $i$ can turn into $C$ whereas $B$ should not be able to turn into $C$. For the same reason calling it $C$ is not valid either.

Perhaps this example deserves closer attention. We name this type of phenomenon the ``delayed choice'' phenomenon, to emphasize the point that when $A$ becomes $i$, although it has committed to becoming either $B$ or $C$ instead of $D$, it has delayed the choice of whether to become $B$ or $C$ until the final reaction takes place.  
This is the same phenomenon occurring in the first example given when discussing \hbox{(weak-)trace} equivalence.
Neither \hbox{(weak-)trace} equivalence nor bisimulation can be applied in systems that exhibit ``delayed choice''.
There are two reasons that the phenomenon is interesting; firstly, there may be a sense in which it is related to the efficiency of the implementation, because the use of delayed choice may allow for a smaller number of intermediate species in implementing the same CRN. Secondly, this phenomenon actually does arise in actual systems, as presented in \cite{Grun}.

We note an important distinction between the various notions of
equivalence discussed here and those found in the Petri net
literature.  Whereas two Petri nets are compared for
(reachability/trace/bisimulation)-equivalence for a particular initial
state~\cite{PEM99}, we are concerned about the various notions of
equivalence of two CRNs for \textit{all} initial states.  This
distinction may limit the applicability of common verification
methodologies and software tools~\cite{Holzmann1997,BM2010}, since the
set of initial states is by necessity always infinite (and the set of
reachable states from a particular initial state may also be
infinite).  Finally, we note that \cite{Lakin} proposes yet another
notion of equivalence based on serializability from database and
concurrency theory.  The serializability result works on a class of
implementations that are ``modular''.  Formal reactions are
\textit{encoded} by a set of implementation reactions and species.
Roughly speaking, modular implementations ensure that each formal
reaction has a unique and correct encoding that does not
``cross-talk'' with the encodings of other formal reactions.  In
general, this results in a one-to-one mapping between formal reactions
and their encodings.  Implementation CRNs satisfying the formal
modularity definitions of \cite{Lakin} will correctly emulate their
target CRN.  However, this class of implementation CRNs precludes those
that utilize ``delayed choice''.  Interestingly, when restricted to
``modular'' implementations, the notion of serializability and our
notion of pathway decomposition have a close
correspondence.

Our approach (originally developed in \cite{S11}) differs from any of the above in that we ignore the target CRN and pay attention only to the implementation CRN. Namely, we simply try to infer what CRN the given implementation would look like in a hypothetical world where we cannot observe the intermediate species. We call this notion ``formal basis.'' We show that not only is the formal basis unique for any valid implementation, but it also has the convenient property that a CRN that does not have any intermediate species has itself as its formal basis. This leads us to a simple definition of CRN equivalence; we can declare two CRNs to be equivalent if and only if they have the same formal basis. Therefore, unlike trace equivalence or weak bisimulation \cite{Qing,Johnson2016}, our definition is actually an equivalence relation and therefore even allows for the comparison of an implementation with another implementation.

\section{Theory}

\subsection{Overview}\label{section:theory_overview}

In previous sections we saw that a reaction which is a single step in the target CRN gets implemented as a pathway of reactions which involves intermediate species whose net effect only changes the number of ``formal'' species molecules. For instance, the pathway $A\to i, i+B\to j, j\to C+k, k\to D$ involves intermediate molecules $i$, $j$, and $k$ but the net effect of this pathway is to consume $A$ and $B$ and produce $C$ and $D$. In this sense this pathway may be viewed as an implementation of $A+B\to C+D$.

In contrast, we will not want to consider the pathway $A\to i, i\to B, B\to j, j\to C$ to be an implementation of $A\to C$, even though its net effect is to consume $A$ and produce $C$. Intuitively, the reason is that this pathway, rather than being an indivisible unit, looks like a composition of smaller unit pathways each implementing $A\to B$ and $B\to C$.

The core idea of our definition, which we call \textit{pathway decomposition}, is to identify all the pathways which act as indivisible units in the above sense. The set of these ``indivisible units'' is called the formal basis of the given CRN. If we can show that all potential pathways in the CRN can be expressed as compositions of these indivisible units, then that will give us ground to claim that this formal basis may be thought of as the target CRN that the given CRN is implementing.

\subsection{Basic definitions}\label{section:definitions}
The theory of pathway decomposition will be developed with respect to a chosen set $\mathbb F$ of species called the {\bf formal species}; all other species will be {\bf intermediate species}.  All the definitions and theorems below should be implicitly taken to be with respect to the choice of $\mathbb F$. As a convenient convention, we use upper case and lower case letters to denote formal and intermediate chemical species, respectively.

\begin{definition}
A \textbf{state} is a multiset of species. If every species in a state $S$ is a formal species, then $S$ is called a \textbf{formal state}. In this paper we will use $+$ and $-$ to denote multiset sum and multiset difference respectively, e.g., $S+T$ will denote the sum of two states $S$ and $T$.
\end{definition}
\begin{definition}
If $S$ is a state, $\mbox{Formal}(S)$ denotes the multiset we obtain by removing all the intermediate species from $S$.
\end{definition}
\begin{definition}
A \textbf{reaction} is a pair of multisets of species $(R, P)$ and it is \textbf{trivial} if $R=P$. Here, $R$ is called the set of \textbf{reactants} and $P$ is called the set of \textbf{products}. We say that the reaction $(R,P)$ \textbf{can occur} in the state $S$ if $R\subseteq S$. If both $R$ and $P$ are formal states, then $(R,P)$ is called a \textbf{formal reaction}. If $r=(R,P)$, we will sometimes use the notation $\bar r$ to denote the reverse reaction $(P,R)$.
\end{definition}
\begin{definition}
If $(R,P)$ is a reaction that can occur in the state $S$, we write $S\oplus (R,P)$ to denote the resulting state $S- R+ P$. As an operator, $\oplus$ is left-associative.
\end{definition}
\begin{definition}
A \textbf{CRN} is a (nonempty) set of nontrivial reactions. A CRN that contains only formal reactions is called a \textbf{formal CRN}.
\end{definition}
\begin{definition}
A \textbf{pathway} $p$ of a CRN $\mathcal C$ is a (finite) sequence of reactions $(r_1,\ldots, r_k)$ with $r_i\in \mathcal C$ for all $i$. We say that a pathway can occur in the state $S$ if all its reactions can occur in succession starting from $S$. Note that given any pathway, we can find a unique minimal state from which the pathway can occur. We will call such state the \textbf{minimal initial state}, or simply the \textbf{initial state} of the pathway. Correspondingly, the \textbf{final state} of a pathway will denote the state $S\oplus r_1\oplus r_2\oplus\cdots\oplus r_k$ where $S$ is the (minimal) initial state of the pathway. If both the initial and final states of a pathway are formal, but not necessarily the intermediate states, it is called a \textbf{formal pathway}. A pathway is called \textbf{trivial} if its initial state equals its final state. In this paper, we will write $p+q$ to denote the concatenation of two pathways $p$ and $q$.
\end{definition}

To absorb these definitions, we can briefly study some examples. Consider the chemical reaction $2A+B\to C$. According to our definitions, this will be written $(\llbrace A,A,B\rrbrace,\llbrace C\rrbrace)$. Here, $\llbrace A,A,B\rrbrace$ is called the reactants and $\llbrace C\rrbrace$ is called the products, just as one would expect. Note that this reaction can occur in the state $\llbrace A,A,A,B,B\rrbrace$ but cannot occur in the state $\llbrace A,B,C,C,C,C\rrbrace$ because the latter state does not have all the required reactants. If the reaction takes place in the former state, then the resulting state will be $\llbrace A,B,C\rrbrace$ and thus we can write $\llbrace A,A,A,B,B\rrbrace\oplus(\llbrace A,A,B\rrbrace,\llbrace C\rrbrace)=\llbrace A,B,C\rrbrace$. In this paper, although we formally define a reaction to be a pair of multisets, we will interchangeably use the chemical notation whenever it is more convenient. For instance, we will often write $2A+B\to C$ instead of $(\llbrace A,A,B\rrbrace,\llbrace C\rrbrace)$.

Note that we say that a pathway $p=(r_1,r_2,\ldots,r_k)$ can occur in the state $S$ if $r_1$ can occur in $S$, $r_2$ can occur in $S\oplus r_1$, $r_3$ can occur in $S\oplus r_1\oplus r_2$, and so on. For example, consider the pathway that consists of $2A+B\to C$ and $B+C\to A$. This pathway cannot occur in the state $\llbrace A,A,B\rrbrace$ because even though the first reaction can occur in that state, the resulting state after the first reaction, which is $\llbrace C\rrbrace$, will not have all the reactants required for the second reaction to occur. In contrast, it is easy to see that this pathway can occur in the state $\llbrace A,A,B,B\rrbrace$, which also happens to be its minimal initial state.

We also point out that we cannot directly express a reversible reaction in this formalism. Thus, a reversible reaction will be expressed using two independent reactions corresponding to each direction, e.g., $A\rightleftharpoons B$ will be expressed as two reactions: $A\to B$ and $B\to A$.

Before we proceed, we formally define the notion of tidiness which we informally introduced in \mbox{Section \ref{section:introduction}}.

\begin{definition}
Let $p$ be a pathway with a formal initial state and $T$ its final state. Then, a (possibly empty) pathway $p'=(r_1,\ldots,r_k)$ is said to be a \textbf{closing pathway} of $p$ if $p'$ can occur in $T$ and $T\oplus r_1\oplus\cdots\oplus r_k$ is a formal state. A CRN is \textbf{weakly tidy} if every pathway with a formal initial state has a closing pathway.
\end{definition}
As was informally explained before, this means that the given CRN is
always capable of cleaning up all the intermediate species. For
example, the CRN $\{A\to i,\ i+B\to C\}$ will not be weakly tidy because
if the system starts from the state $\llbrace A\rrbrace$, it can transition to the
state $\llbrace i\rrbrace$ and become ``stuck'' in a non-formal state: there does
not exist a reaction to convert the intermediate species $i$ back into
some formal species.

For a more subtle example, let us consider the CRN $\{A\to i+B,\ i+B\to B\}$, which is weakly tidy according to the definition as stated above. In fact, it is easy to see that this implementation CRN will never get stuck when it is operating by itself, starting with any formal initial state. However, this becomes problematic when we begin to think about composing different CRNs. Namely, when intermediate species require other formal species in order to get removed, the implementation CRN may not work correctly if some other formal reactions are also operating in the system. For instance, if the above implementation runs in an environment that also contains the reaction $B\to C$, then it is no longer true that the system is always able to get back to a formal state.

This is not ideal because the ability to compose different CRNs, at
least in the case where they do not share any intermediate species, is
essential for CRN implementations to be useful. To allow for this type
of composition, and more importantly to allow for the proofs of Theorems in Section \ref{sec:modular} and to make the algorithm defined in Section \ref{sec:algorithm} tractable, we define a stronger notion of tidiness which is preserved under such composition.

\begin{definition}
A closing pathway is \textbf{strong} if its reactions do not consume any formal species.
A CRN is \textbf{strongly tidy} if every pathway with a formal initial state has a strong closing pathway.
\end{definition}

In the rest of the paper, unless indicated otherwise, we will simply say tidiness to mean strong tidiness. Similarly, we will simply say closing pathway to mean strong closing pathway. For some examples of different levels of tidiness, see Figure \ref{fig:tidiness}.

\begin{figure}[!h]
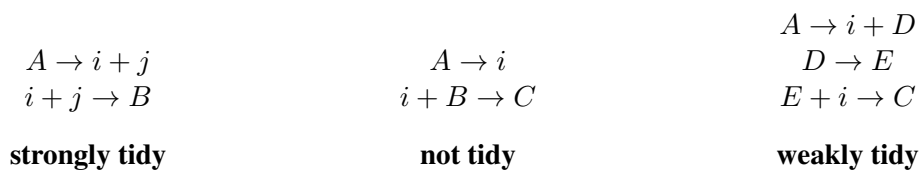

\centering
\vspace{0.7cm}
\begin{minipage}[b]{0.3\linewidth}
\centering
$A\to i + j$\\
$i+j\to B$ \\
\vspace{0.3cm}
\textbf{strongly tidy}\\
\end{minipage}
\begin{minipage}[b]{0.3\linewidth}
\centering
$A\to i$\\
$i+B\to C$\\
\vspace{0.3cm}
\textbf{not tidy}\\
\end{minipage}
\begin{minipage}[b]{0.3\linewidth}
\centering
$A\to i + D$\\
$D\to E$\\
$E+i\to C$\\
\vspace{0.3cm}
\textbf{weakly tidy}\\
\end{minipage}
\caption{Some examples of tidy and non-tidy CRNs}\label{fig:tidiness}
\end{figure}

\subsection{Pathway decomposition}\label{section:theory_pd}

Now we formally define the notion of pathway decomposition. Following our intuition from Section \ref{section:theory_overview}, we first define what it means to implement a formal reaction.

\begin{definition}\label{def:regularity}
Consider a pathway $p=(r_1,\ldots,r_k)$ and let $S_i= S\oplus r_1\oplus\cdots\oplus r_i$, so that $S_0, S_1,\ldots, S_k$ are all the states that $p$ goes through.  Then, $p$ is  \textbf{regular} if there  exists a \textbf{turning point} reaction $r_j=(R',P')$ such that $\mbox{Formal}(S_i)\subseteq S$ for all $i<j$, $\mbox{Formal}(S_i)\subseteq T$ for all $i\geq j$, and $\mbox{Formal}(S_{j-1}-R')=\emptyset$. 
\end{definition}

\begin{definition}\label{def:implement}
We say that a pathway $p=(r_1,\ldots,r_k)$ \textbf{implements} a formal reaction $(R,P)$ if it is regular and $R$ and $P$ are equal to the initial and final states of $p$, respectively.
\end{definition}

While the first condition is self-evident, the second condition needs a careful explanation. It asserts that there should be a point in the pathway prior to which we only see the formal species from the initial state and after which we only see the formal species from the final state. The existence of such a ``turning point'' allows us to interpret the pathway as an implementation of the formal reaction $(R,P)$ where in a sense the real transition is occurring at that turning point.
Importantly, this condition rules out such counterintuitive implementations as $(A\to i,\ i\to C+j,\ C+j\to k,\ k\to B)$ or $(A\to i+B,\ i+B\to j+A,\ j+A\to B)$ as implementations of $A\to B$. 
%
Note that a formal pathway that consumes but does not produce formal species prior to its turning point, and thereafter produces but does not consume formal species, is by this definition regular, and this is the ``typical case.''  However our definition also allows additional flexibility; for example, the reactants can fleetingly bind, as $B$ does in the second and third reactions of $(A \to i, i+B \to j, j \to B+i, i+B \to C)$, whose turning point is unambiguously the last reaction. One may also wonder why we need the condition $\mbox{Formal}(S_{j-1}-R') = \emptyset$. This is to prevent ambiguity that may arise in the case of catalytic reactions. Consider the pathway $(A\to i+A, i\to B)$. Without the above condition, both reactions in this pathway qualify as a turning point, but the second reaction being interpreted as the turning point is counterintuitive because the product $A$ gets produced before the turning point. 

One problem of the above definition is that it interprets the pathway $(A\to i,\ A\to i,\ i\to B,\ i\to B)$ as implementing $A+A\to B+B$. As explained in Section \ref{section:theory_overview}, we would like to be able to identify such a pathway as a composition of smaller units.

\begin{definition}
We say that a pathway $p$ can be \textbf{partitioned} into two pathways $p_1$ and $p_2$ if $p_1$ and $p_2$ are subsequences of $p$ (which need not be contiguous, but must preserve order) and every reaction in $p$ belongs to exactly one of $p_1$ and $p_2$. Equivalently, we can say $p$ is formed by \textbf{interleaving} $p_1$ and $p_2$.
\end{definition}

\begin{definition}
A formal pathway $p$ is \textbf{decomposable} if $p$ can be partitioned into $p_1$ and $p_2$ that are each formal pathways. A nonempty formal pathway that is not decomposable is called \textbf{prime}.
\end{definition}
For example, consider the formal pathway $p=(A\to i,\ B\to i,\ i\to C,\ i\to D,\ D\to j,\ j\to E)$. This pathway is not prime because it can be decomposed into two formal pathways $p_1=(A\to i,\ i\to C)$ and $p_2=(B\to i,\ i\to D,\ D\to j,\ j\to E)$. Note that within each of the two subsequences, reactions must appear in the same order as in the original pathway $p$. In this example, $p_1$ is already a prime pathway after the first decomposition, whereas $p_2$ can be further decomposed into $(B\to i,\ i\to D)$ and $(D\to j,\ j\to E)$. In this manner, any nonempty formal pathway can eventually be decomposed into one or more prime pathways. Note that such a decomposition may not be unique, e.g., $p$ can also be decomposed into $(A\to i,\ i\to D)$, $(B\to i,\ i\to C)$, and $(D\to j,\ j\to E)$.
  
\begin{figure}[!b]
\centering
\vspace{0.5cm}
\includegraphics[width=\textwidth]{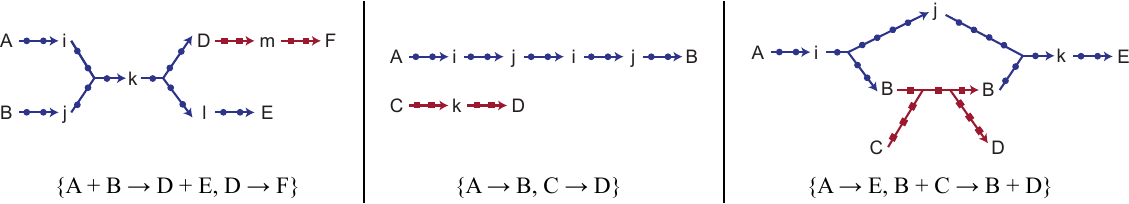}
\caption {Three examples of decomposable formal pathways and the formal bases of their corresponding CRNs. The partition of reactions is marked by lines of different types and colors.  In the right most example, the decomposed pathway denoted by blue lines with circles (which shows up as $A\to E$ in the formal basis) is not regular, and therefore pathway decomposition equivalence does not apply.}\label{fig:decomposability}
\end{figure}

\begin{definition}
The set of prime pathways in a given CRN is called the \textbf{elementary basis} of the CRN. The \textbf{formal basis} is the set of $(\mbox {initial state}, \mbox{final state})$ pairs of the pathways in the elementary basis.
\end{definition}

Note that the elementary basis and/or the formal basis can be either finite or infinite. The elementary basis may contain trivial pathways, and the formal basis may contain trivial reactions.

\begin{definition}
A CRN is \textbf{regular} if every prime pathway implements some formal reaction (in particular, it must have a well-defined turning point reaction as defined in Definition \ref{def:regularity}). Equivalently, a CRN is regular if every prime pathway is regular.
\end{definition}

\begin{definition}
Two tidy and regular CRNs are said to be \textbf{pathway decomposition equivalent} if their formal bases are identical, up to addition or removal of trivial reactions.
\end{definition}

For clarity, we remind the reader here that the definitions and theorems in this section are implicitly taken to be with respect to the choice of $\mathbb F$. In particular, this means that each choice of $\mathbb F$ gives rise to a different pathway decomposition equivalence relation. For instance, CRNs $\{A\to i,\ i\to C\}$ and $\{A\to C\}$ are clearly pathway decomposition equivalent with respect to the conventional choice of $\mathbb F$, which contains exactly those species named with upper case letters, but if e.g.~we defined $\mathbb F'=\mathbb F\cup\{i\}$, these two CRNs would not be pathway decomposition equivalent with respect to $\mathbb F'$.

\subsection{Theorems}\label{section:theory_theorems}
\subsubsection{Properties}
It is almost immediate that pathway decomposition equivalence satisfies many nice properties, some of which are expressed in the following theorems. 
\begin{theorem}\label{thm:equivalence}
For any fixed choice of $\mathbb F$, pathway decomposition equivalence with respect to $\mathbb F$ is an equivalence relation, i.e., it satisfies the reflexive, symmetric, and transitive properties.
\end{theorem}
\begin{theorem}
If $\mathcal C$ is a formal CRN, its formal basis is itself.
\end{theorem}
\begin{corollary}
If $\mathcal C_1$ and $\mathcal C_2$ are formal CRNs, they are pathway decomposition equivalent if and only if $\mathcal C_1=\mathcal C_2$, up to removal or addition of trivial reactions.
\end{corollary}
\begin{theorem}
Any formal pathway of $\mathcal C$ can be generated by interleaving one or more prime pathways of $\mathcal C$.
\end{theorem}

It is perhaps worth noting here that the decomposition of a formal
pathway may not always be unique. For example, the pathway $(A\to i,\
B\to i,\ i\to C,\ i\to D)$ can be decomposed in two different ways:
$(A\to i,\ i\to C)$ and $(B\to i,\ i\to D)$, and $(A\to i,\ i\to D)$
and $(B\to i,\ i\to C)$. Pathway decomposition differs from other
notions such as (weak) bisimulation or (weak) trace equivalence in
that it allows such degeneracy of interpretations. We note that such
degeneracy, which is closely related to the previously mentioned
delayed choice phenomenon, may permit a more efficient
implementation of a target CRN in terms of the number of species or
reactions used in the implementation CRN. For example, if we wish to
implement the formal CRN consisting of the twelve reactions $A\rightleftharpoons B$,
$A\rightleftharpoons C$, $A\rightleftharpoons D$, $B\rightleftharpoons
C$, $B\rightleftharpoons D$ and $C\rightleftharpoons D$, it may be
more efficient to implement it as the following eight reactions: $A\rightleftharpoons i$,
$B\rightleftharpoons i$, $C\rightleftharpoons i$ and
$D\rightleftharpoons i$.

The following theorems illuminate the relationship between a tidy and regular CRN $\mathcal C$ and its formal basis $\mathcal F$ and how to better understand this degeneracy of interpretations.
\begin{definition}
Let $\mathcal C$ be a tidy and regular CRN and $\mathcal F$ its formal basis. Suppose $p=(r_1,\ldots,r_k)$ is a formal pathway in $\mathcal C$ (i.e.~$r_i\in \mathcal C$ for all $i$) and $q=(s_1,\ldots,s_l)$ is a formal pathway in $\mathcal F$ (i.e.~$s_i\in \mathcal F$ for all $i$). Then, we say $p$ can be \textbf{interpreted} as $q$ if
\begin{enumerate}
\item $q$ can occur in the initial state $S$ of $p$,
\item $S\oplus r_1\oplus\cdots\oplus r_k=S\oplus s_1\oplus\cdots\oplus s_l$, and
\item there is a decomposition of $p$ such that if we replace a selected turning point reaction of each prime pathway 
with the corresponding reaction of $\mathcal F$ and remove all other reactions, the result is $q$.
\end{enumerate}
\end{definition}
It is clear that the interpretation may not be unique, because there can be many different decompositions of $p$ as well as many different choices of the turning point reactions. For example, consider the pathway $p=(A\to i,\ B\to j,\ i\to C,\ j\to D)$, which has a unique decomposition into pathways $(A\to i,\ i\to C)$ and $(B\to j,\ j\to D)$. In each of these constituent pathways, there are two ways to select a turning point reaction. If we picked $A\to i$ and $j\to D$, the process in condition 3 would yield $(A\to C,\ B\to D)$ as the interpretation of $p$. On the other hand, if we selected $i\to C$ and $B\to j$ as our turning point reactions, $p$ would end up being interpreted as $(B\to D,\ A\to C)$.

One might also wonder why we do not simply require that $p$ and $q$ must have the same initial states. This is because of a subtlety in the concept of the minimal initial state, which arises due to a potential parallelism in the implementation. For instance, consider the pathway $(A\to i,\ B\to A,\ i\to B)$. This pathway, which can be interpreted as two formal reactions $A\to B$ and $B\to A$ occuring in parallel, has initial state $\llbrace A,B\rrbrace$. However, no such parallelism is allowed in the formal CRN and thus this pathway is forced to correspond to either $(A\to B,\ B\to A)$ or $(B\to A,\ A\to B)$, neither of which has initial state $\llbrace A,B\rrbrace$.

\begin{theorem} \label{thm:interpretation}
Suppose $\mathcal C$ is a tidy and regular CRN and $\mathcal F$ is its formal basis.
\begin{enumerate}
\item For any formal pathway $q$ in $\mathcal F$, there exists a formal pathway $p$ in $\mathcal C$ whose initial and final states are equal to those of $q$, such that $p$ can be interpreted as $q$.
\item Any formal pathway $p$ in $\mathcal C$ can be interpreted as some pathway $q$ in $\mathcal F$.
\end{enumerate}
\end{theorem}
\begin{proof}
\begin{enumerate}
\item Replace each reaction in $q$ with the corresponding prime pathway of $\mathcal C$.
\item Fix a decomposition of $p$ and pick a turning point for each prime pathway. Replace the turning points with the corresponding formal basis reaction and remove all other reactions. We call the resulting pathway $q$. Then it suffices to show that $q$ can occur in the initial state $S$ of $p$. We show this by a hybrid argument. 

Define $p_j$ to be the pathway obtained by replacing the first $j$ turning points in $p$ by the corresponding formal basis elements and removing all other reactions that belong to those prime pathways. In particular, note that $p_0=p$ and $p_l=q$. We show that $p_{j}$ can occur in the initial state of $p_{j-1}$ for all $j>0$. First, write $p_{j-1}=(r_1,\ldots,r_m)$ and $p_{j}=(r_{i_1},\ldots,r_{i_k},s_j,r_{i_{k+1}},\ldots,r_{i_n})$. 
Then it follows from the definition of a turning point that $\text{Formal}(S\oplus r_{i_1}\oplus\cdots\oplus r_{i_{x-1}})\supseteq\text{Formal}(S\oplus r_1\oplus\cdots\oplus r_{i_x-1})$ for every $1\leq x\leq k$. Therefore $(r_{i_1},\ldots,r_{i_k})$ can occur in $S$. (Note that we need not worry about the intermediate species because $(r_{i_1},\ldots,r_{i_k})$ has a formal initial state.) Moreover, since the definition of a turning point asserts that all the reactants must be consumed at the turning point, it also implies that $\text{Formal}(S\oplus r_{i_1}\oplus\cdots\oplus r_{i_k})\supseteq\text{Formal}(S\oplus r_1\oplus\cdots\oplus r_{t-1} - X + R)$ where $r_t=(X,Y)$ denotes the turning point that is being replaced by $s_j$ in this round and $R$ denotes the reactants of $s_j$. Therefore, $s_j$ can occur in $S\oplus r_{i_1}\oplus \cdots\oplus r_{i_k}$. Finally, it again follows from the definition of a turning point that $\text{Formal}(S\oplus r_{i_1}\oplus\cdots\oplus r_{i_{k}}\oplus s_j\oplus r_{i_{k+1}}\oplus\cdots\oplus r_{i_{x-1}})\supseteq\text{Formal}(S\oplus r_1\oplus\cdots\oplus r_{i_x-1})$ for every $k+1\leq x\leq n$. We conclude that $(r_{i_1},\ldots,r_{i_k},s_j,r_{i_{k+1}},\ldots,r_{i_n})=p'_j$ can occur in $S$.
\end{enumerate}
\end{proof}

It is interesting to observe that tidiness is not actually used in the proof of Theorem \ref{thm:interpretation} above (nor in that of Theorem \ref{thm:reachability} below), so that condition could be removed from the theorem statement.  We retain the tidiness condition to emphasize that this is when the theorem characterizes the behavior of the CRN; without tidiness, a CRN could have many relevant behaviors that take place along pathways that never return to a formal state, and these behaviors would not be represented in its formal basis.

In our final theorem we prove that pathway decomposition equivalence implies formal state reachability equivalence. Note that the converse is not true because $\{A\to B,B\to C,C\to A\}$ is not pathway decomposition equivalent to $\{A\to C,C\to B,B\to A\}$.
\begin{theorem} \label{thm:reachability}
If two tidy and regular CRNs $\mathcal C_1$ and $\mathcal C_2$ are pathway decomposition equivalent, they give rise to the same reachability between formal states.
\end{theorem}
\begin{proof}
Suppose formal state $T$ is reachable from formal state $S$ in $\mathcal C_1$, i.e.~there is a formal pathway $p$ in $\mathcal C_1$ whose initial state is $S$ and final state is $T$. By Theorem \ref{thm:interpretation}, it can be interpreted as some pathway $q$ consisting of the reactions in the formal basis of $\mathcal C_1$. Since $\mathcal C_1$ and $\mathcal C_2$ have the same formal basis, by another application of Theorem \ref{thm:interpretation}, there exists some formal pathway $p'$ in $\mathcal C_2$ that can be interpreted as $q$. That is, the initial and final states of $p'$ are $S$ and $T$ respectively, which implies that $T$ is reachable from $S$ in $\mathcal C_2$ also. By symmetry between $\mathcal C_1$ and $\mathcal C_2$, the theorem follows.
\end{proof}

\subsubsection{Modular composition of CRNs}\label{sec:modular}
As we briefly mentioned in Section \ref{section:definitions}, it is very important for the usefulness of a CRN implementation that it be able to be safely composed with other CRNs. For instance, consider the simplest experimental setup of putting the molecules of the implementation CRN in the test tube and measuring the concentration of each species over time. In practice, the concentration measurement of species $A$ is typically carried out by implementing a catalytic reaction that uses $A$ to produce fluorescent material. Therefore even this simple scenario already involves a composition of two CRNs, namely the implementation CRN itself and the CRN consisting of the measurement reactions. It is evident that the ability to compose CRNs would become even more essential in more advanced applications.

 In this section, we prove theorems that show that pathway decomposition equivalence is preserved under composition of CRNs, as long as those CRNs do not share any intermediate species.
\begin{theorem}\label{thm:compose_tidiness}
Let $\mathcal C$ and $\mathcal C'$ be two CRNs that do not share any intermediate species. 
Then, $\mathcal C \cup \mathcal C'$ is tidy if and only if both $\mathcal C$ and $\mathcal C'$ are tidy.
\end{theorem}
\begin{proof}
For the forward direction, by symmetry it suffices to show that $\mathcal C$ is tidy. We begin by proving the following lemma.
\begin{lemma}
Let $\mathcal C$ and $\mathcal C'$ be two CRNs that do not share any intermediate species. Let $p$ be any formal pathway in $\mathcal C\cup \mathcal C'$. If we partition $p$ into two pathways $p_1$ and $p_2$ such that $p_1$ is a pathway of $\mathcal C$ and $p_2$ is a pathway of $\mathcal C'$, then each of $p_1$ and $p_2$ is formal.
\end{lemma}
\begin{proof}
Since $\mathcal C$ and $\mathcal C'$ do not share any intermediate species, it follows that all the intermediate species in the initial state of $p_1$ will also show up in the initial state of $p$ and all the intermediate species in the final state of $p_1$ will also show up in the final state of $p$. Hence $p_1$ must be formal. The case for $p_2$ follows by symmetry.
\end{proof}
Now let $p$ be any pathway in $\mathcal C$ with a formal initial state. Since $p$ is also a pathway in $\mathcal C\cup \mathcal C'$, it has a closing pathway $q$ in $\mathcal C\cup \mathcal C'$. Since $s=p+q$ is a formal pathway, we can partition it into $s_1$ and $s_2$ as in the above lemma. In particular, since all the reactions in $p$ belong to $\mathcal C$, we have $s_1=p+q_1$ and $s_2=q_2$ where $q_1$ and $q_2$ are a partition of $q$ such that $q_1$ is a pathway of $\mathcal C$ and $q_2$ is a pathway of $\mathcal C'$. Since $s_1$ is formal by the lemma, $q_1$ is a closing pathway of $p$. Hence, $\mathcal C$ is tidy.

For the reverse direction, suppose $p$ is a pathway of $\mathcal C\cup \mathcal C'$ that has a formal initial state. Since $\mathcal C$ and $\mathcal C'$ do not share intermediate species, we can partition the intermediate species found in the final state of $p$ into two multisets $A$ and $A'$, corresponding to the intermediate species used by $\mathcal C$ and $\mathcal C'$ respectively. Now, if we remove from $p$ all the reactions that belong to $\mathcal C'$ and call the resulting pathway $q$, then the multiset of all the intermediate species found in the final state of $q$ will be exactly $A$. This is because the removed reactions, which belonged to $C'$, cannot consume or produce any intermediate species used by $\mathcal C$. Since $\mathcal C$ is tidy, $q$ has a closing pathway $r$. This time, remove from $p$ all the reactions that belong to $\mathcal C$ and call the resulting pathway $q'$. By a symmetric argument, $q'$ must have a closing pathway $r'$. Now observe that $r+r'$ is a closing pathway for $p$.
\end{proof}
\begin{theorem}\label{thm:regularity}
Let $\mathcal C$ and $\mathcal C'$ be two CRNs that do not share any intermediate species. Then, $\mathcal C \cup \mathcal C'$ is regular if and only if both $\mathcal C$ and $\mathcal C'$ are regular.
\end{theorem}
\begin{proof}
For the forward direction, simply observe that any prime pathway $p$ of $\mathcal C$ is also a prime pathway of $\mathcal C\cup \mathcal C'$ and therefore must be regular. Hence, $\mathcal C$ is regular. By symmetry, $\mathcal C'$ is also regular.

For the reverse direction, let $p$ be a prime pathway in $\mathcal C \cup \mathcal C'$. Partition $p$ into two subsequences $q$ and $q'$, which contains all reactions of $p$ which came from $\mathcal C$ and $\mathcal C'$ respectively. Since the two CRNs do not share any intermediate species, it is clear that $q$ and $q'$ must both be formal. Since $p$ was prime, it implies that one of $q$ and $q'$ must be empty. Therefore, $p$ is indeed a prime pathway in either $\mathcal C$ or $\mathcal C'$, and since each was a regular CRN, $p$ must be regular.
\end{proof}
\begin{theorem}\label{thm:compose}
Let $\mathcal C$ and $\mathcal C'$ be two tidy and regular CRNs that do not share any intermediate species, and $\mathcal F$ and $\mathcal F'$ their formal bases respectively. Then the formal basis of $\mathcal C \cup \mathcal C'$ is exactly $\mathcal F\cup \mathcal F'$.
\end{theorem}
\begin{proof}
Let $p$ be a prime pathway in $\mathcal C \cup \mathcal C'$. By the same argument as in the proof of Theorem \ref{thm:regularity}, $p$ is a prime pathway of either $\mathcal C$ or $\mathcal C'$. Therefore, the formal basis of $\mathcal C \cup \mathcal C'$ is a subset of $\mathcal F\cup \mathcal F'$. The other direction is trivial.
\end{proof}

We note that the ability to compose CRNs has another interesting consequence. Frequently, molecular implementations of CRNs involve intermediate species that are specific to a pathway implementing a particular reaction, such that intermediates that belong to pathways that implement different reactions do not react with each other. This is a strong constraint on the architecture of the implementations that can facilitate their verification. (For instance, it has been observed and used by Lakin et al.~in \cite{Lakin}.) We observe that in such cases Theorem \ref{thm:compose} provides an easier way to find the formal basis of the implementation CRN. Namely, we can partition the CRN into disjoint subsets that do not share intermediate species with one another, find the formal basis of each subset, and then take the union of the found formal bases. For example, if the implementation CRN was $\{A\to i,\ i\to B,\ A\to j,\ j\to C,\ j+C\to k,\ k\to D\}$, then it can be partitioned into CRNs $\{A\to i,\ i\to B\}$ and $\{A\to j,\ j\to C,\ j+C\to k,\ k\to D\}$ such that they do not share intermediate species with each other. It is straightforward to see that the formal bases of these two subsets are $\{A\to B\}$ and $\{A\to C,\ A+C\to D\}$ respectively, so the formal basis of the whole implementation CRN must be $\{A\to B,\ A\to C,\ A+C\to D\}$. Similarly, Theorems \ref{thm:compose_tidiness} and \ref{thm:regularity} ensure that we can test for tidiness and regularity of the implementation CRN by testing tidiness and regularity of each of these subsets. 

\section{Algorithm} \label{sec:algorithm}
In this section, we present a simple algorithm for finding the formal basis of a given CRN. The algorithm can also test tidiness and regularity.

Our algorithm works by enumerating pathways that have formal initial states.
The running time of our algorithm depends on a quantity called maximum width, which can be thought of as the size of the largest state that a prime pathway can ever generate. Unfortunately it is easy to see that this quantity is generally unbounded; e.g., $\{A\to i, i\to i+i, i\to\emptyset \}$ has a finite formal basis $\{A\to \emptyset\}$ but it can generate arbitrarily large states.\footnote{Clearly, there may  also be cases where the formal basis itself is infinite, e.g.~$\{A\to i, i\to i+i, i\to B\}$.} However, since such implementations are highly unlikely to arise in practice, in this paper we focus on the bounded width case. We note that even in the bounded width case it is still nontrivial to come up with an algorithm that finishes in finite time, because it is unclear at what width we can safely stop the enumeration.

\subsection{Exploiting bounded width}
We begin by introducing a few more definitions and theorems.

\begin{definition}
A pathway that has a formal initial state is called \textbf{semiformal}.
\end{definition}
\begin{definition}
A semiformal pathway $p$ is \textbf{decomposable} if $p$ can be partitioned into two nonempty subsequences (which need not be contiguous) that are each semiformal pathways.

It is obvious that this reduces to our previous definition of decomposability if $p$ is a formal pathway.
\end{definition}

\begin{definition}
Let $p=(r_1,\ldots,r_k)$ be a pathway and let $S_i=S\oplus r_1\oplus \cdots \oplus r_i$ where $S$ is the initial state of $p$. The \textbf{width} of $p$ is defined to be $\max_i |S_i|$.
\end{definition}
\begin{definition}
The \textbf{branching factor} of a CRN $\mathcal C$ is defined to be the following value.
\[
\max_{(R,P)\in \mathcal C}{\max\{|R|, |P|\}}
\]
We note that many implementations that arise in practice have small branching factors (e.g.~\cite{SSW10, C09, QSW11}).
\end{definition}

\begin{theorem}\label{prop:lemma}
Suppose that pathway $p$ is obtained by interleaving pathways $p_1,\ldots,p_l$. Let $S$ be the initial state of $p$ and $S_1,\ldots,S_l$ the initial states of $p_1,\ldots,p_l$ respectively. Then, $S\subseteq S_1+S_2+\cdots+S_l$.
\end{theorem}
\begin{theorem}\label{thm:main}
If $p$ is an undecomposable semiformal pathway of width $w>0$, there exists an undecomposable semiformal pathway of width smaller than $w$ but at least $(w - b)/b$, where $b$ is the branching factor of the CRN. (Note that if $w$ is small, the lower bound $(w-b)/b$ might be negative. In this case, it would simply mean that there exists an undecomposable semiformal pathway of width $0$, which would be the empty pathway.)
\end{theorem}
\begin{proof}
Since $w>0$, $p$ is nonempty. Let $p_{-1}$ denote the pathway obtained by removing the last reaction $(R,P)$ from $p$. Also, let $S_0,\ldots, S_k$ be the states that $p$ goes through, and $S'_0,\ldots,S'_{k-1}$ the states that $p_{-1}$ goes through. $S_i$ is potentially unequal to $S_i'$ because if the last reaction in $p$ consumes some new formal species, then the minimal initial state of $p_{-1}$ might be smaller than that of $p$. 

It is obvious that the minimal initial state of $p_{-1}$ is smaller than the minimal initial state of $p$ by at most $|R|$, i.e., $|S_0|-|S'_0|\leq |R|$. This means that for all $0\leq i\leq k-1$, we have that $|S_i|-|S'_i|\leq |R|$. Clearly, if there exists some $0\leq i\leq k-1$ such that $|S_i|=w$, then $|S'_i|\geq |S_i|-|R|=w-|R|\geq w-b$, so $p_{-1}$ has width at least $w-b$. If there exists no such $i$, then we have that $|S_k|=w$. Clearly, $|S_{k-1}|=|S_k|-|P|+|R|$ and it follows that
\[
|S_k|-|P|+|R|-|S'_{k-1}|=|S_{k-1}|-|S'_{k-1}|\leq |R|.
\]
This is equivalent to $|S_k|-|S'_{k-1}|\leq |P|$. Since $|S_k|=w$, we have that $|S'_{k-1}|\geq w-|P| \geq w-b$. Thus, $p_{-1}$ achieves width at least $w-b$.

Then, we decompose $p_{-1}$ until it is no longer decomposable. As a result, we will end up with $l\geq 1$ undecomposable pathways $p_1, p_2,\ldots, p_l$ which by interleaving can generate $p_{-1}$. Also, they are all semiformal. First, we show that $l$ is at most $b$. Assume towards a contradiction that $l>b$. Then, by the pigeonhole principle, there exists $i$ such that $(R-\mbox{Formal}(R),P)$ can occur in the sum of the final states of $p_1,\ldots,p_{i-1},p_{i+1},\ldots,p_l$ (since $|R-\mbox{Formal}(R)|\leq b$ and $(R-\mbox{Formal}(R),P)$ can occur in the sum of the final states of $p_1,\ldots,p_l$, the at most $b$ reactants of $(R-\mbox{Formal}(R),P)$ are distributed among $l>b$ pathways and there exists at least one $p_i$ that does not provide a reactant and can be omitted).
 Then, consider the decomposition $(p_i, p'_i)$ of $p_{-1}$ where $p'_i$ denotes the pathway we obtain by interleaving $p_1,\ldots,p_{i-1},p_{i+1},\ldots,p_l$ in the same order that those reactions occur in $p_{-1}$. By Theorem \ref{prop:lemma}, $p'_i$ is semiformal. Since $p_j$'s are all semiformal, this means that the intermediate species in the final state of $p'_i$ will be exactly the same as those in the sum of the final states of $p_1,\ldots,p_{i-1},p_{i+1},\ldots,p_l$. That is, the final state of $p'_i$ contains all the intermediate species that $(R,P)$ needs to occur, i.e., $p'_i$ with $(R,P)$ appended at the end should have a formal initial state. However, this means that $p$ is decomposable which is a contradiction. Hence, $l\leq b$.

Now, note that if we have $l$ pathways each with widths $w_1, \ldots, w_l$, any pathway obtained by interleaving them can have width at most $\sum_{i=1}^l w_i$. Since $p_{-1}$ had width  at least $w-b$, we have that $w-b\leq \sum_{i=1}^l w_i$. Then, if $w_i < (w-b)/b$ for all $i$, then $\sum_{i=1}^l w_i < w-b$, which is contradiction. Thus, we conclude that at least one of $p_1,\ldots,p_l$ has width greater than or equal to $(w-b)/b$. It is also clear that its width cannot exceed $w$. Thus, we have found a pathway $p'$ which
\begin{enumerate}
\item has a smaller length than $p$, and
\item has width at least $(w-b)/b$ and at most $w$.
\end{enumerate}
If $p'$ has width exactly $w$, then we have failed to meet the requirements of the claim. However, since we have decreased the length of the pathway by at least one, and the width of a zero-length pathway is always 0, we can eventually get a smaller width than $w$ by repeating this argument.  The first time that the width decreases, we will have found a pathway $p'$ that satisfies the theorem statement, because in that case conditions (1) and (2) must hold by the arguments above.
\end{proof}
\begin{corollary}\label{cor:main}
Suppose $w$ and $w_{\text{max}}$ are integers such that $(w+1)b\leq w_{\text{max}}$. Then, if there is no undecomposable semiformal pathway of width greater than $w$ and less than or equal to $w_{\text{max}}$, then there exists no undecomposable semiformal pathway of width greater than $w$.
\end{corollary}
\begin{proof}
Assume towards a contradiction that there exists an undecomposable semiformal pathway $p$ of width $w'>w$. If $w'\leq w_{\text{max}}$, then it is an immediate contradiction. Thus, assume that $w'>w_{\text{max}}$. By Theorem \ref{thm:main}, we can find a smaller undecomposable semiformal pathway $q$ of width $v$ where $(w'-b)/b\leq v<w'$. Since $w'>w_{\text{max}}$, we have that $v\geq (w'-b)/b > (w_{\text{max}}-b)/b \geq ((w+1)b-b)/b = w$. If $v\leq w_{\text{max}}$, we have a contradiction. If $v> w_{\text{max}}$, then take $q$ as our new $p$ and repeat the above argument. Since $v$ is smaller than $w'$ by at least one, we will eventually reach a contradiction.

Thus, there exists no undecomposable semiformal pathway of width greater than $w$.
\end{proof}

\subsection{Overview}
While Corollary \ref{cor:main} gives us a way to exploit the bounded width assumption, it is still unclear whether the enumeration can be made finite, because the number of undecomposable semiformal pathways of bounded width may still be infinite. For an easy example, if the CRN consists of $\{A\to i,\ i\to j,\ j\to i,\ j\to B\}$, we have infinitely many undecomposable semiformal pathways of width $1$, because after the initial reaction $A\to i$, the segment $i\to j,\ j\to i$ can be repeated arbitrarily many times without ever making the pathway decomposable. In this section, we sketch at high level how this difficulty is resolved in our finite-time algorithm.

The principal technique that lets us avoid infinite enumeration of pathways is memoization. To use memoization, we first define what is called the \textbf{signature} of a pathway, which is a collection of information about many important properties of the pathway, such as its initial and final states, decomposability, etc. It turns out that the number of possible signatures of bounded width pathways is always finite, even if the number of pathways themselves may be infinite. This means that the enumeration algorithm does not need to duplicate pathways with the same signatures, provided the signatures alone give us sufficient information for determining the formal basis and for testing tidiness and regularity of the CRN.

Therefore, the algorithm consists in enumerating all semiformal pathways of width up to ${(w+1)b}$, where $w$ is the maximum width of the undecomposable semiformal pathways discovered so far, while excluding pathways that have the same signatures as previously discovered pathways. It is important to emphasize that no a priori knowledge of the width bound is assumed, and the algorithm is guaranteed to halt as long as there exists some finite bound.
While the existence of this algorithm shows that the problem of finding the formal basis is decidable with the bounded width assumption, the worst-case time complexity seems to be adverse as is usual for algorithms based on exhaustive search. It is an open question to understand the computational complexity of this problem as well as to find an algorithm that has better practical performance. Another important open question is whether the problem without the bounded width assumption is decidable.

\subsection{Signature of a pathway}
While Corollary \ref{cor:main} gives us a way to make use of the bounded width assumption, it is still unclear whether the enumeration can be made finite, because the number of undecomposable semiformal pathways of bounded width may still be infinite. To resolve this problem, we need to define a few more concepts.

\begin{definition}
Let $p$ be a semiformal pathway. The \textbf{decomposed final states (DFS)} of $p$ is defined as the set of all unordered pairs $(T_1,T_2)$ that can be obtained by decomposing $p$ into two semiformal pathways and taking their final states. Note that for an undecomposable pathway, the DFS is the empty set.
\end{definition}
\begin{definition}
Let $p=(r_1,\ldots,r_k)$ be a semiformal pathway. Also, let $S_i=S\oplus r_1\oplus\cdots\oplus r_i$ where $S$ is the initial state of $p$. The \textbf{formal closure} of $p$ is defined as the unique minimal state $S'$ such that $\mbox{Formal}(S_i)\subseteq S'$ for all $i$.
\end{definition}
\begin{definition}
The \textbf{regular final states (RFS)} of $p$ is defined as the set of all minimal states $T$ such that there exists a potential turning point reaction $r_{j}=(R,P)$ which satisfies $\mbox{Formal}(S_i)\subseteq S$ for all $i< j$, $\mbox{Formal}(S_i)\subseteq T$ for all $i\geq j$, and $\mbox{Formal}(S_{j-1}-R)=\emptyset$.
\end{definition}

Some explanation is in order.  Although the RFS definition applies equally to semiformal pathways that are or could be regular, and to semiformal pathways that are not and cannot be regular, the RFS provides a notion of ``what the final state would/could be if the pathway were regular''. As examples, first consider the semiformal pathway $( A \to i, B+i \to j, j \to X + k)$.  The second and third reactions are potential turning points, and the RFS is $\{ \llbrace X\rrbrace \}$. One can easily check that if the pathway were to be completed in a way that its final state does not contain $X$, the resulting pathway cannot be regular (e.g.~were it to be closed by $X+k\to Y$, the pathway becomes irregular). Now consider $(A \to i, i \to B+j, B+j \to k)$. Only the first two reactions are potential turning points, and the RFS is $\{ \llbrace B\rrbrace \}$. One can also check in this case that the only way that this semiformal pathway can be completed as a regular pathway is for it to have $B$ in its final state.
 Finally consider $(A \to i, i \to A+j, A+j \to B)$, which is in fact a regular formal pathway implementing $A \to B$.  Because every reaction is a potential turning point by our definition, the RFS is $\{\llbrace A,B\rrbrace, \llbrace B\rrbrace\}$.  One of these states is the actual final state, corresponding to the actual turning point, and therefore we can see that this formal pathway is regular.

\begin{definition}
The \textbf{signature} of the pathway is defined to be the $6$-tuple of the initial state, final state, width, formal closure, DFS, and RFS.
\end{definition}

\begin{theorem}\label{thm:enumt}
If $m$ is any finite number, the set of signatures of all semiformal pathways of width up to $m$ is finite.
\end{theorem}
\begin{proof}
Clearly, there is only a finite number of possible initial states, final states, widths, formal closures, and RFS. Also, since there is only a finite number of possible final states, there is only a finite number of possibilities for DFS.
\end{proof}

\begin{theorem}\label{thm:enumc}
Suppose $p_1$ and $p_2$ are two pathways with the same signature. Then, for any reaction $r$, $p_1+(r)$ and $p_2+(r)$ also have the same signature.
\end{theorem}
\begin{proof}
Let $p_1'=p_1+(r)$ and $p_2'=p_2+(r)$. It is trivial that $p_1'$ and $p_2'$ have the same initial and final states, formal closure, and width.

First, we show that $p_1'$ and $p_2'$ have the same DFS.
Suppose $(T_1, T_2)$ is in the DFS of $p_1'$. That is, there exists a decomposition $(q_1', q_2')$ of $p_1'$ where $q_1'$ and $q_2'$ have final states $T_1$ and $T_2$. The last reaction $r$ is either contained in $q_1'$ or $q_2'$. Without loss of generality, suppose the latter is the case. Then, if $q_1=q_1'$ and $q_2+(r)=q_2$, then $(q_1, q_2)$ should decompose $p_1$, which is a prefix of $p_1'$.  Since $p_1$ and $p_2$ have the same DFS, there should be a decomposition $(s_1, s_2)$ of $p_2$ that has the same final states as $q_1$ and $q_2$. Clearly, $(s_1, s_2+(r))$ should be a decomposition of $p_2'$ and thus $(T_1, T_2)$ is also in the DFS of $p_2'$. By symmetry, it follows that $p_1'$ and $p_2'$ have the same DFS.

Now we argue that $p_1'$ and $p_2'$ should have the same RFS. Suppose $T$ is contained in the RFS of $p_1'$.
%
\begin{enumerate}
%
\item If the potential turning point for $T$ in $p_1'$ is the last reaction $r$, with $r = (R,P)$, then it must be the case that $T = \mbox{Formal}(P)$.  Because $p_2$ has the same formal closure and final state as $p_1$, which was sufficient to ensure that $r$ was a valid potential turning point in $p_1'$, $r$ will also be a valid potential turning point in $p_2'$.  Consequently, $T$ is also in the RFS of $p_2'$. 
%
%
\item Otherwise, the potential turning point reaction for $T$ in $p_1'$, call it $t = (R, P)$, also appears in $p_1$.  
Since the initial state of $p_1$ must be a subset of the initial state of $p_1'$, $t$ is also a potential turning point for $p_1$.  Since ``midway through'' the potential turning point reaction, all formal species must be gone, we conclude that in fact $p_1$ and $p_1'$ have the same initial state.  That is, $R$ contains no formal species that aren't already in the final state of $p_1$.  Thus, all shared states after $t$ are the same, and some subset $T'$ of $T$ must be contained in the RFS of $p_1$. By assumption, $T'$ is also in the RFS of $p_2$, and $p_2$ has the same final state as $p_1$.  Since $R$ contains no formal species that aren't already in the final state of $p_2$, the initial states of $p_2$ and $p_2'$ are the same.  Consequently, the potential turning point of $p_2$ corresponding to $T'$ is also a potential turning point for $p_2'$.  This ensures that $T$ is in the RFS for $p_2'$. 
\end{enumerate}
\end{proof}

\begin{theorem}
A nonempty pathway $p$ is a prime pathway if and only if its signature satisfies the following conditions:
\begin{enumerate}
\item The initial and final states are formal.
\item The DFS is the empty set.
\end{enumerate}
\end{theorem}

\subsection{Algorithm for enumerating signatures}

It is now clear that we can find the formal basis by enumerating the signatures of all undecomposable semiformal pathways. In this section we present a simple algorithm for achieving this.

\begin{verbatim}
function enumerate(p, w, ret)
    if p is not semiformal or has width greater than w then return ret
    sig = signature of p
    if sig is in ret then return ret
    add sig to ret
    for every reaction rxn
        ret = enumerate(p + [rxn], w, ret)
    end for
    return ret
end function

function main()
    w_max = 0
    b = branching factor of the given CRN
    while true
        signatures = enumerate([], w_max, {})
        w = maximum width of an undecomposable pathway in signatures
        if (w+1)*b <= w_max then break
        w_max = (w+1)*b
    end while
    return signatures
end function
\end{verbatim}

The subroutine \texttt{enumerate} is a function that enumerates the signatures of all semiformal pathways of width at most \texttt{w}. Note that it uses memoization to avoid duplicating pathways that have identical signatures, as justified by Theorem \ref{thm:enumc}. Because of this memoization, Theorem \ref{thm:enumt} ensures that this subroutine will terminate in finite time.

The subroutine \texttt{main} repeatedly calls \texttt{enumerate}, increasing the width bound according to Corollary \ref{cor:main}. It is obvious that \texttt{main} will terminate in finite time if and only if there exists a bound to the width of an undecomposable semiformal pathway.

It is out of scope of this paper to attempt theoretical performance analysis of this algorithm or to study the computational complexity of finding the formal basis. While there are obvious further optimizations by which the performance of the above algorithm can be improved, we meet our goal of this paper in demonstrating the existence of a finite time algorithm and leave further explorations as a future task.

\subsection{Testing tidiness and regularity}
Finally, we discuss how to use the enumerated signatures to test tidiness and regularity of the given CRN.
\begin{theorem}\label{thm:tidiness2}
A CRN is tidy if and only if every undecomposable semiformal pathway has a closing pathway.
\end{theorem}
\begin{proof}
The forward direction is trivial. For the reverse direction, we show that if a CRN is not tidy, there exists an undecomposable semiformal pathway that does not have a closing pathway.

By definition, there exists a semiformal pathway $p$ that does not have a closing pathway.  Consider a minimal-length example of such a pathway.  If $p$ is undecomposable, then we are done.  So suppose that $p$ is decomposable into two semiformal pathways $p_1$ and $p_2$.  By the minimality of $p$, both pathways $p_1$ and $p_2$ must have closing pathways. However, since the final state of $p$ has the same intermediate species as the sum of the final states of $p_1$ and $p_2$ (by Theorem 4.1 and the fact that $p_1$ and $p_2$ are semiformal), the two closing pathways concatenated will be a closing pathway of $p$ (because a closing pathway does not consume any formal species).  This contradicts that $p$ does not have a closing pathway, and thus we conclude that the case where $p$ is decomposable is impossible.
\end{proof}
\begin{theorem}\label{thm:tidiness3}
Let $p$ be an undecomposable semiformal pathway that has a closing pathway. Then $p$ also has a closing pathway $q$ such that $p+q$ is undecomposable.
\end{theorem}
\begin{proof}
Let $q$ be a minimal-length closing pathway for $p$.  Note that $p+q$ is a formal pathway.  If $p+q$ is undecomposable, we are done.  So suppose that $p+q$ decomposes into two formal pathways $p_1$ and $p_2$, which by definition must both be nonempty.  Then it must be the case that one of $p_1$ or $p_2$ contains all the reactions of $p$, because otherwise $p$ must be decomposable as well.  Without loss of generality, suppose $p_1$ contains all the reactions of $p$.  Then $p_2$ consists only of reactions from $q$.  This means that the reactions of $q$ that went into $p_1$ constitute a shorter closing pathway $qÕ$ for $p$, contradicting the minimality of $q$.  We conclude that $p+q$ must have been undecomposable.
\end{proof}

To test tidiness, we attempt to find a closing pathway for each undecomposable semiformal pathway $p$ enumerated by the main algorithm. Theorem \ref{thm:tidiness2} ensures that it suffices to consider only these pathways. We do this by enumerating the signatures of all semiformal pathways of the form $p+q$ where $q$ is a pathway that does not consume a formal species, but only those of width up to $w_{\text{max}}$ ($w_{\text{max}}$ is the maximum width of the undecomposable semiformal pathways discovered by the main algorithm). Theorem \ref{thm:tidiness3} ensures that it is safe to enforce this width bound. 

The testing of regularity is trivial, using the following theorem.
\begin{theorem}
A prime pathway is regular if and only if its RFS contains its final state.
\end{theorem}
\begin{proof}
The potential turning point corresponding to the final state proves regularity.  If the final state is lacking in the RFS, then none of the potential turning points qualify as a turning point and the pathway is not regular.
\end{proof}

We emphasize that these methods work only because of the bounded width assumption we made on undecomposable semiformal pathways. Without this assumption, it is unclear whether these problems still remain decidable.

\subsection{Optimization techniques}\label{sec:optimization}
In this section, we discuss some optimization techniques that can be used to improve the performance of the main enumeration algorithm. While we provide no theoretical analysis of these techniques, we report that there are test instances on which these techniques speed up the algorithm by many orders of magnitude.

\begin{definition}
If $S$ is a state, $\textbf{Intermediate}(S)$ denotes the multiset that consists of exactly all the intermediate species in $S$.
\end{definition}
\begin{theorem}\label{thm:initialmain}
If $p$ is an undecomposable semiformal pathway of CRN $\mathcal C$ with an initial state of size $m>0$, there exists an undecomposable semiformal pathway of $\mathcal C$ with an initial state of size smaller than $m$ but at least $\min_{(R,P)\in \mathcal C}\{(m - |\mbox{Formal}(R)|) / |\mbox{Intermediate}(R)|\}$.
\end{theorem}
\begin{proof}
Since $m>0$, $p$ is nonempty. Let $p_{-1}$ denote the pathway obtained by removing the last reaction $(R,P)$ from $p$. Let $x$ be the number of formal species in $R$. Also, let $S$ and $S_{-1}$ denote the initial states of $p$ and $p_{-1}$ respectively.

It is obvious that the initial state of $p_{-1}$ is smaller than the initial state of $p$ by at most $x$, i.e., $|S_{-1}|\geq |S|-x$. Then, we decompose $p_{-1}$ until it is no longer decomposable. As a result, we will end up with $l$ undecomposable pathways $p_1, p_2,\ldots, p_l$ which by interleaving can generate $p_{-1}$. Also, they are all semiformal. First, we show that $l$ is at most $y$, where $y$ is the number of intermediate species in $R$ (clearly, $x+y=|R|$). Assume towards a contradiction that $l>y$. Then, by the pigeonhole principle, there exists $i$ such that $(\mbox{Intermediate}(R),P)$ can occur in the sum of the final states of $p_1,\ldots,p_{i-1},p_{i+1},\ldots,p_l$ (as in the proof of Theorem \ref{thm:main}). Then, consider the decomposition $(p_i, p'_i)$ of $p_{-1}$ where $p'_i$ denotes the pathway we obtain by interleaving $p_1,\ldots,p_{i-1},p_{i+1},\ldots,p_l$ in the same order that those reactions occur in $p_{-1}$. By Theorem \ref{prop:lemma}, $p'_i$ is semiformal. Since $p_j$'s are all semiformal, this means that the intermediate species in the final state of $p'_i$ will be exactly the same as those in the sum of the final state of $p_1,\ldots,p_{i-1},p_{i+1},\ldots,p_l$. 
That is, the final state of $p'_i$ contains all the intermediate species that $(\mbox{Intermediate}(R),P)$ needs to occur, i.e., $p'_i$ with $(R,P)$ appended at the end should have a formal initial state. However, this means that $p$ is decomposable which is a contradiction. Hence, $l\leq y$.

Now, note that if we have $l$ semiformal pathways whose initial states have size $m_1, \ldots, m_l$, any pathway obtained by interleaving them can have an initial state of size at most $\sum_{i=1}^l m_i$. Since $p_{-1}$ had an initial state of size  at least $m-x$ and $l\leq y$, we can conclude that at least one of $p_1,\ldots,p_l$ has an initial state of size at least $(m-x)/y$. It is also clear that the size of its initial state  cannot exceed $m$. Thus, we have found a pathway $p'$ which
\begin{enumerate}
\item has a smaller length than $p$, and
\item has an initial state of size at least $\min_{(R,P)\in \mathcal C}\{(m - |\mbox{Formal}(R)|) / |\mbox{Intermediate}(R)|\}$ and at most $m$.
\end{enumerate}
If $p'$ has an initial state of size exactly $m$, then we have failed to meet the requirements of the claim. However, since we have decreased the length of the pathway by at least one and the initial state of a zero-length pathway is of size $0$, we can eventually get an initial state of size smaller than $m$ by repeating this process. The first time that the size of the initial state decreases, we will have found a pathway $p'$ that satisfies the theorem statement because in that case conditions (1) and (2) must hold by the arguments above.
\end{proof}
The above theorem allows us to maintain a bound \texttt{i\_max} on the size of initial states during enumeration, in a similar manner to how \texttt{w\_max} is maintained. Since our enumeration algorithm is essentially brute-force, imposing this additional bound may significantly reduce the number of pathways that need to be enumerated.

Moreover, the proof of the above theorem immediately lets us optimize the constants in Theorem \ref{thm:main}.
\begin{theorem}\label{thm:opt}
If $p$ is an undecomposable semiformal pathway of width $w>0$, there exists an undecomposable semiformal pathway of width smaller than $w$ but at least $(w - b)/b_r$, where $$b_r=\max_{(R,P)\in\mathcal C}\{|\mbox{Intermediate}(R)|\}.$$
\end{theorem}
\begin{proof}
Same as the proof of Theorem \ref{thm:main}, except that we argue $l\leq b_r$ instead of $l\leq b$, using the argument from the proof of Theorem \ref{thm:initialmain}.
\end{proof}

The following theorem helps us further eliminate a huge number of pathways from consideration.
\begin{definition}
Let $p$ be a semiformal pathway. We say that $p$ is \textbf{strongly decomposable} if $p$ can be decomposed into two semiformal pathways $p_1$ and $p_2$ such that at least one of $p_1$ and $p_2$ is a formal pathway.
\end{definition}
\begin{theorem}\label{thm:sdec}
Let $p$ be a semiformal pathway. If it is strongly decomposable, any semiformal pathway that contains $p$ as a prefix is decomposable.
\end{theorem}
\begin{proof}
Suppose $p$ is strongly decomposable into formal pathway $p_1$ and semiformal pathway $p_2$. We show that for any $p'$, if $p+p'$ is semiformal, then it is decomposable into $p_1$ and $p_2+p'$. It suffices to show that $p_2+p'$ is semiformal. Assume towards a contradiction that the initial state of $p_2+p'$ contains an intermediate species. Since $p_2$ is semiformal, it means that there is an intermediate species $x$ contained in $S-T$, where $T$ is the final state of $p_2$ and $S$ is the initial state of $p'$. Let $T'$ be the final state of $p$. Since $p_1$ is formal, $\mbox{Intermediate}(T')=\mbox{Intermediate}(T)$. Hence, $x$ is also contained in $S-T'$, which means that $x$ appears in the initial state of $p+p'$. This is a contradiction to our initial assumption that $p+p'$ was semiformal. Hence, $p_2+p'$ is semiformal and therefore $p+p'$ is decomposable.
\end{proof}

Finally, we show that CRNs that possess a certain structure can be processed extremely quickly. Since most published implementations do have this structure (e.g.~\cite{SSW10, C09, QSW11}), this observation is very useful in practice.
\begin{definition}
Let $\mathcal C$ be a CRN. We define the following two sets, which partition the set of intermediate species of $\mathcal C$ according to whether they ever participate in a reaction as a reactant.
\begin{align*}
W(\mathcal C)&=\{\text{species $x$ in $\mathcal C$}: \text{$x$ is not formal and $x$ never appears as a reactant in reactions of $\mathcal C$}\}\\
NW(\mathcal C)&=\{\text{species $x$ in $\mathcal C$}: \text{$x$ is not formal and $x\notin W(\mathcal C)$}\}
\end{align*}
Moreover, for any state $S$, we will denote by $S^{NW(\mathcal C)}$ the multiset containing exactly those species of $S$ that belong to $NW(\mathcal C)$.
\end{definition}
\begin{definition}
A CRN $\mathcal C$ is said to have \textbf{monomolecular substructure} if for every reaction $(R,P)\in\mathcal C$, both $|R^{NW(\mathcal C)}|$ and $|P^{NW(\mathcal C)}|$ are at most one.
\end{definition}
\begin{theorem}
Let $\mathcal C$ be a CRN that has monomolecular substructure and $p=(r_1,\ldots,r_k)$ any undecomposable semiformal pathway of $\mathcal C$. Also, let $S_0$ be the initial state of $p$ and $S_i=S_0\oplus r_1\oplus\cdots\oplus r_i$ all the states that $p$ goes through. Then, for any $0\leq i\leq k$, we have $|S_i^{NW(\mathcal C)}|\leq 1$.
\end{theorem}
\begin{proof}
We prove this by induction on $k$. If $k=1$, the claim holds trivially. Now assume that the claim holds for all pathways of length up to $k-1$. Let $p_{-1}$ be the pathway obtained by removing the last reaction $r_k$ from $p$. Note that $r_k$  consumes up to one intermediate species because the CRN has monomolecular substructure. Moreover, $r_k$ must consume at least one intermediate species because otherwise $p$ can be decomposed into $(r_1,\ldots,r_{k-1})$ and $(r_k)$. Therefore $r_k$ consumes exactly one intermediate species $x$, which by definition must be in $NW(\mathcal C)$. By induction hypothesis, $|S_i^{NW(\mathcal C)}|\leq 1$ for all $0\leq i\leq k-1$ (the initial state of $p_{-1}$ and the initial state of $p$ differ only by formal species) and in particular the final state of $p_{-1}$ contains at most one intermediate species that belongs to $NW(\mathcal C)$. This implies that this intermediate species must be $x$, because otherwise $p=p_{-1}+(r_k)$ would not be semiformal. The last reaction $r_k$ consumes this $x$ and produces at most one intermediate species that belongs to $NW(\mathcal C)$, which means that $|S_k^{NW(\mathcal C)}|\leq 1$. The theorem now follows by induction.
\end{proof}
The above theorem implies that when we run our algorithm on a CRN with monomolecular substructure, there is no need to enumerate semiformal pathways that ever go through a state that contains more than one species from $NW(\mathcal C)$.

\subsection{Testing pathway decomposition equivalence}
In this section, we have presented an algorithm for enumerating the formal basis of a given CRN, which is guaranteed to halt if there is a finite bound to the width of an undecomposable semiformal pathway. Moreover, this algorithm can also be used to test whether the CRN is tidy and regular.

Hence, we are finally in a position to be able to verify the correctness of CRN implementations; namely, using the above algorithm we can test whether the target CRN and the implementation CRN are pathway decomposition equivalent. Since it immediately follows from definition that the target CRN is tidy and regular and that its formal basis is equal to itself, this verification amounts to checking that the implementation CRN is tidy and regular and that its formal basis is equal to the target CRN up to addition or removal of trivial reactions. All of these tasks can easily be achieved using our algorithm.

We note that because of Theorem~\ref{thm:equivalence} our theory applies also to the more general scenario of comparing two arbitrary CRNs. In this case, one would need to enumerate the elementary and formal bases of both CRNs, verify that both CRNs are tidy and regular, and finally check that their formal bases are identical up to addition or removal of trivial reactions.

\section{Handling the general case}\label{section:discussion}\label{section:fuel}
In this section, we discuss some important issues that pertain to practical applications and hint at the possibility of further theoretical investigations.

As we briefly mentioned earlier, many CRN implementations that arise in practice involve not only formal and intermediate species but also what are called fuel and waste species. Fuel species are chemical species that are assumed to be always present in the system at fixed concentration, as in a buffer. For instance, DNA implementations \cite{SSW10, C09, QSW11} often employ fuel species that are present in the system in large concentrations and have the ability to transform formal species into various other intermediates. This type of ``implementation'' is also prevalent in biological systems, where the concentrations of energy-carrying species such as ATP, synthetic precursors such as NTPs, and general-purpose enzymes such as ribosomes and polymerases, are all maintained in roughly constant levels by the cellular metabolism.

In CRN verification, the standard approach to fuel species is to preprocess implementation CRNs such that all occurrences of fuel species are simply removed. For instance, if the CRN contained reaction $A+g\to i+t$ where $g$ and $t$ are fuel species, the preprocessed CRN will only have $A\to i$. The justification for this type of preprocessing is that since fuel species are always present in the system in large concentrations by definition, consuming or producing a finite number of fuel species molecules do not have any effect on the system. 
In particular, it can be shown that holding fuel species at constant concentration, versus simply removing them from the reactions while appropriately adjusting the reaction rate constants, leads to exactly the same mass-action ODE's and continuous-time Markov chains.

On the other hand, implementations sometimes produce ``waste'' species as byproducts. Waste species are supposed to be chemically inert and thus cannot have interaction with other formal or intermediate species. However, in practice it is often difficult to implement a chemical species which is completely inert and therefore they may interact with other species in trivial or nontrivial ways. Therefore the main challenge is to ensure that these unwanted interactions do not give rise to a logically erroneous behavior. One way to deal with this problem is to first verify that such waste species are indeed ``effectively inert'' and then preprocess them in a similar manner to fuel species. To achieve this we need to answer two important questions: first, how to satisfactorily define ``effectively inert'' and second, how such waste species may be algorithmically identified.

Another related problem which must be solved before we can use pathway decomposition is that some implementations may have multiple chemical species that are interpreted as the same formal species. (For example, see DNA implementations \cite{SSW10, C09} with ``history domains.'' An example is given in Section \ref{sec:case}.) Since our mathematical framework implicitly assumes one-to-one correspondence between formal species of the target CRN and formal species of the implementation CRN, it is not immediately clear how we can apply our theory in such cases.

Interestingly, the weak bisimulation-based approach to CRN equivalence proposed in \cite{Qing,Johnson2016} does not seem to suffer from any of these problems, because it in fact does not make a particular distinction between these different types of species except fuel species.
Rather, it requires that there must be a way to interpret each species that appears in the implementation CRN as one or more formal species. For instance, if $\{A\rightleftharpoons i, B+i\rightleftharpoons j, j\to C\}$ is proposed as an implementation of $A+B\to C$, the weak bisimulation approach will interpret $A$ and $i$ as $\llbrace A\rrbrace$, $B$ as $\llbrace B\rrbrace$, $j$ as $\llbrace A,B\rrbrace$, and $C$ as $\llbrace C\rrbrace$. Therefore the state of the system at any moment will have an instantaneous interpretation as some formal state, which is not provided by pathway decomposition. On the other hand, the weak bisimulation approach cannot handle interesting phenomena that are allowed in the pathway decomposition approach, most notably the delayed choice phenomenon explained in Section \ref{section:motivation}.

Our proposed solution to the problem of wastes and multiple formal labeling is a compositional hybrid approach between weak bisimulation and pathway decomposition. Namely, we take the implementation CRN from which only the fuel species have been preprocessed, and tag as ``formal'' species all the species that have been labeled by the user as either an implementation of a target CRN species or a waste. All other species are tagged as ``intermediates''. Then we can apply the theory of pathway decomposition to find its formal basis (with respect to the tagging, as opposed to the smaller set of species in the target CRN). Note that waste species must be tagged as ``formal'' rather than ``intermediate'' because they will typically accumulate, and thus tagging them as ``intermediate'' would result in a non-tidy CRN to which pathway decomposition theory does not apply.
Finally, we verify that the resulting formal basis of tagged species is weak bisimulation equivalent to the target CRN under the natural interpretation, which interprets implementations of each target CRN species as the target CRN species itself and wastes as ``null.''   If the implementation is incorrect, or if some species was incorrectly tagged as ``waste'', the weak bisimulation test will fail.  See Figure \ref{fig:combined} for example.

\begin{figure}[!h]
\centering
\vspace{0.7cm}
\begin{minipage}[b]{0.3\linewidth}
\centering
$A_1\to i$\\
$i\to B_1+W$\\
$A_2\to j$\\
$j\to B_2$\\
$W+j\to B_1$\\
\ \\
\vspace{0.3cm}
\textbf{Implementation CRN}
\end{minipage}
\begin{minipage}[b]{0.3\linewidth}
\centering
$A_1 \to B_1 +W$\\
$A_2 \to B_2$\\
$A_2+W\to B_1$\\
\ \\
\ \\
\vspace{0.3cm}
\textbf{Formal basis}
\end{minipage}
\begin{minipage}[b]{0.3\linewidth}
\centering
$A\to B$\\
\ \\
\ \\
\ \\
\vspace{0.3cm}
\textbf{Under weak bisimulation}
\end{minipage}
\caption{The compositional hybrid approach for verifying an implementation of the formal CRN $\{A\to B\}$. We first apply pathway decomposition, treating the upper case species as formal species and lower case species as intermediate species. Then, we apply weak bisimulation using the natural interpretation which interprets $A_1$ and $A_2$ as $\{| A|\}$, $B_1$ and $B_2$ as $\{| B|\}$, and $W$ as $\emptyset$. Thus, in two steps, the implementation CRN has been shown to be a correct implementation of $\{A\to B\}$.}\label{fig:combined}
\end{figure}

On the other hand, we note that the weak bisimulation approach can sometimes handle interesting cases which pathway decomposition cannot. For instance, the design proposed in \cite{QSW11} for reversible reactions implements $A+B\rightleftharpoons C+D$ as $\{A\rightleftharpoons i, i+B\rightleftharpoons j, j\rightleftharpoons k+C, k\rightleftharpoons D\}$. Note that this implementation CRN is not regular according to our theory because of the prime pathway $A\to i, i+B\to j, j\to k+C, k+C\to j, j\to i+B, i\to A$. Interestingly, this type of design seems to directly oppose the foundational principles of the pathway decomposition approach. One of the key ideas that inspired pathway decomposition is that of ``base touching,'' namely the idea that even though the evolution of the system involves many intermediate species, a pathway implementing a formal reaction must eventually produce all its formal products and thus ``touch the base.'' This principle is conspicuously violated in the above pathway, because while the only intuitive way to interpret it is as $A+B\to C+D$ and then $C+D\to A+B$, the first part does not touch the base by producing a $D$ molecule. In contrast, the weak bisimulation approach naturally has no problem handling this implementation: $i$ is interpreted as $\llbrace A\rrbrace$, $j$ is interpreted as $\llbrace A,B\rrbrace$, and $k$ is interpreted as $\llbrace D\rrbrace$.

The fact that the two approaches are good for different types of instances motivates us to further generalize the compositional hybrid approach explained above. To define the generalized compositional hybrid approach, we begin by formally introducing the weak bisimulation approach of \cite{Qing,Johnson2016}. As we have seen above, the weak bisimulation approach requires an ``interpretation map'' $m$ from species of the implementation CRN to states of the target CRN. For instance, in the above example $m$ was defined as $m(A)=m(i)=\llbrace A\rrbrace$, $m(B)=\llbrace B\rrbrace$, $m(j)=\llbrace A,B\rrbrace$, $m(C)=\llbrace C\rrbrace$, and $m(D)=m(k)=\llbrace D\rrbrace$. Although the domain of $m$ is technically species of the implementation CRN, there is an obvious sense in which we can also apply it to states, reactions, or pathways. Thus when convenient we will abuse notation to mean $m(S)=\sum_{x\in S} m(x)$ for a state $S$, $m(r)=(m(R),m(P))$ for a reaction $r=(R,P)$, and $m(p)=(m(r_1),m(r_2),\ldots,m(r_k))$ for a pathway $p=(r_1,\ldots,r_k)$. Then, the following definition and theorem are adapted from \cite{Qing,Johnson2016} to fit our definitions of chemical reactions and pathways.
\begin{definition}
{(Section 3.2 of \cite{Johnson2016})}
A target CRN $\mathcal C_1$ and an implementation CRN $\mathcal C_2$ are \textbf{weak bisimulation equivalent} under interpretation $m$ if
\begin{enumerate}
\item for any state $S$ in $\mathcal C_1$, there exists a state $S'$ in $\mathcal C_2$ such that $m(S')=S$,
\item for any state $S'$ in $\mathcal C_2$ and $S=m(S')$,
\begin{enumerate}
\item if $r\in\mathcal C_1$ can occur in $S$, then there exists a pathway $p=(s_1,\ldots,s_k)$ in $\mathcal C_2$ such that $S\oplus r=m(S' \oplus s_1\oplus\cdots\oplus s_k)$ and $m(p)$ is equal to $(r)$ up to addition or removal of trivial reactions, and
\item if $r'\in\mathcal C_2$ can occur in $S'$, then $m(r')$ is either a reaction in $\mathcal C_1$ or a trivial reaction, and thus $m(S'\oplus r')=S\oplus m(r')$.
\end{enumerate}
\end{enumerate}
\end{definition}
\begin{theorem}\label{thm:qing}
{(An immediate corollary of Theorem 1 of \cite{Johnson2016}) }
If a target CRN $\mathcal C_1$ and an implementation CRN $\mathcal C_2$ are weak bisimulation equivalent under interpretation $m$, then the following holds:
\begin{enumerate}
\item If $S$ is a state in $\mathcal C_1$, $p$ is a pathway in $\mathcal C_1$ that can occur in $S$, and $S'$ is a state in $\mathcal C_2$ such that $m(S')=S$, then there exists a pathway $p'$ in $\mathcal C_2$ such that $p'$ can occur in $S'$ and $m(p')$ is equal to $p$ up to addition or removal of trivial reactions.
\item If $S'$ is a state in $\mathcal C_2$ and $p'$ is a pathway in $\mathcal C_2$ that can occur in $S'$, then there exists a pathway $p$ in $\mathcal C_1$ such that $p$ can occur in $m(S')$ and $p$ is equal to $m(p')$ up to addition or removal of trivial reactions.
\end{enumerate}
\end{theorem}
Similarly to Theorem \ref{thm:interpretation}, the above theorem establishes a kind of pathway equivalence between the target CRN and the implementation CRN.
Now, we can formally define the generalized compositional hybrid approach as follows.

\begin{definition}\label{def:hybrid}
Suppose we are given a target CRN $\mathcal C_1$ and an implementation CRN $\mathcal C_2$. Let $\mathcal F$ and $\mathcal S$ denote the species of $\mathcal C_1$ and $\mathcal C_2$ respectively. Let $\mathcal X\subseteq S$ be the set of species that have been labeled by the user as implementations of target CRN species or wastes.
In the \textbf{compositional hybrid approach}, we say $\mathcal C_2$ is a \textbf{correct} implementation of $\mathcal C_1$ if there exists some $\mathcal X\subseteq \mathcal V\subseteq S$ such that
\begin{enumerate}
\item $\mathcal C_2$ with respect to $\mathcal V$ as formal species is tidy and regular, and
\item the formal basis of $\mathcal C_2$ with respect to $\mathcal V$ as formal species is weak bisimulation equivalent to $\mathcal C_1$ under some interpretation that respects the labels on $\mathcal X$ provided by the user.
\end{enumerate}
\end{definition}

The flexibility to vary $\mathcal V$ can be useful:  for example, intermediates that are involved in ``delayed choice'' pathways can be kept out of $\mathcal V$ so as to be handled by pathway decomposition, whereas intermediates involved in the aforementioned reversible reaction pathways can be retained within $\mathcal V$ so as to be handled by weak bisimulation.

Finally, we prove a theorem analogous to Theorems \ref{thm:interpretation} and \ref{thm:qing}, in order to provide an intuitive justification for the adequacy of the above definition. We begin by extending the notion of interpretation of pathways that we introduced in Section \ref{section:theory_theorems} to include the concept of interpretation map.
\begin{definition}\label{def:interpretation}
Suppose $\mathcal V$ denotes the set of species of $\mathcal C_2$ that are being tagged as formal species in the compositional hybrid approach. Let $m$ be an interpretation map from $\mathcal V$ to states of $\mathcal C_1$.
We say a formal pathway $p=(r_1,\ldots,r_k)$ in $\mathcal C_2$ can be \textbf{interpreted} as a pathway $q=(s_1,\ldots,s_l)$ in $\mathcal C_1$ under $m$ if
\begin{enumerate}
\item $q$ can occur in $m(S)$, where $S$ is the initial state of $p$,
\item $m(S\oplus r_1\oplus \cdots\oplus r_k)=m(S)\oplus s_1\oplus \cdots \oplus s_l$, and
\item there is a decomposition of $p$ such that if we replace the turning point reaction of each prime pathway with the corresponding element of $\mathcal C_1$ 
(i.e. the corresponding formal basis reaction mapped through $m$)
and remove all other reactions,
then the resulting pathway
is equal to $q$ up to addition or removal of trivial reactions.
\end{enumerate}
\end{definition}

Then, the following theorem provides a sense in which two CRNs that are ``equivalent'' according to the compositional hybrid approach indeed do have equivalent behaviors.

\begin{theorem}\label{thm:hybrid}
Suppose an implementation CRN $\mathcal C_2$ is a correct implementation of the target CRN $\mathcal C_1$ according to the compositional hybrid approach. Then, there exists a mapping $m$ from $\mathcal V$ to states of $\mathcal C_1$ such that the following two conditions hold.
\begin{enumerate}
\item Let $q$ and $S$ be a pathway and a state in $\mathcal C_1$ such that $q$ can occur in $S$. Then, for any state $S'$ in $\mathcal C_2$ that uses species from $\mathcal V$ such that $m(S')=S$, there exists a formal pathway $p$ in $\mathcal C_2$ that can occur in $S'$ and can be interpreted as $q$ under $m$.
\item Any formal pathway $p$ in $\mathcal C_2$ can be interpreted as some pathway $q$ in $\mathcal C_1$ under $m$.
\end{enumerate}
\end{theorem}
\begin{proof}
Let $m$ be the interpretation map provided by the weak bisimulation equivalence \cite{Qing,Johnson2016}, and $\mathcal I$ the formal basis of $\mathcal C_2$ with respect to $\mathcal V$ as formal species.  
\begin{enumerate}
\item By Theorem \ref{thm:qing}, we
  have a pathway $p'$ in $\mathcal I$ that can occur in $S'$ and
  $m(p')$ is equal to $q$ up to addition or removal of trivial reactions. Now replace each reaction in $p'$ by the prime
  pathway that implements that reaction and call the resulting pathway
  $p$. Clearly, $p$ can occur in $S'$. To show that $p$ can be
  interpreted as $q$ under $m$, observe that the first condition of Definition~\ref{def:interpretation} follows from the fact that $q$ can clearly occur in $m(S')$ and $S'$ is a superset of the initial state of $p$ (because $p$ can occur in $S'$). Since $p$ and $p'$ have the same initial and final states and $m(p')=q$, we also satisfy the second condition. The final condition trivially follows from the way $p$ was constructed and the fact that $m(p')$ was equal to $q$ up to addition or removal of trivial reactions.
\item By Theorem \ref{thm:interpretation}, $p$ can be interpreted as some pathway $p'$ in $\mathcal I$. Let $q$ be the pathway we obtain by removing all the trivial reactions from $m(p')$. Now we show that $p$ can be interpreted as $q$ under $m$. For the first condition of Definition~\ref{def:interpretation}, we use Theorem \ref{thm:qing} to see that $q$ can occur in $m(S')$ where $S'$ is the initial state of $p'$. Since $p'$ can occur in the initial state of $p$, this implies that $q$ can also occur in $m(S)$ where $S$ is the initial state of $p$. The second condition follows immediately from the way $q$ was constructed and the fact that $p$ and $p'$ have the same net effect. The final condition follows from the fact that $p$ can be interpreted as $p'$ in $\mathcal I$ and that $m(p')=q$ up to removal of trivial reactions.
\end{enumerate}
\end{proof}

As we shall see in Section \ref{sec:case}, the compositional hybrid approach allows for the verification of interesting real-life systems that neither pathway decomposition nor bisimulation is able to handle individually. In fact, the compositional hybrid approach seems to be the most general approach proposed thus far in terms of the range of implementations that it can address, which, to our best knowledge, includes all currently known enzyme-free DNA implementation techniques. At the same time, we remark that its definition as presented in this paper does not seem to be completely satisfactory. To see why, consider the CRN $\{A\to i,\ i+B_1\leftrightharpoons j_1,\ i+B_2\leftrightharpoons j_2,\ j_1\to C,\ j_2\to C \}$ as an implementation of $\{A+B\to C\}$. Intuitively, it seems that the compositional hybrid approach should have no problem handling this example with $\mathcal V=\{A,B_1,B_2,C\}$ and $m(A)=\llbrace A\rrbrace,\ m(B_1)=m(B_2)=\llbrace B\rrbrace,\ m(C)=\llbrace C\rrbrace$. Surprisingly, it turns out that the prime pathway $A\to i,\ i+B_1\to j_1,\ j_1\to i+B_1,\ i+B_2\to j_2,\ j_2\to C$ is not regular under this choice of $\mathcal V$, because the product $B_1$ is produced before the reactant $B_2$ is consumed. Of course, the compositional hybrid approach can still handle this implementation because we can always choose $\mathcal V=\{A,B_1,B_2,C,i,j_1,j_2\}$ and delegate the whole verification to the bisimulation part. Nonetheless, it is troubling that the above pathway is considered irregular because if indeed $B_1$ and $B_2$ both represent $B$, then there is a sense in which this pathway should really be thought of as $A\to i,\ i+B\to j_1,\ j_1\to i+B,\ i+B\to j_2,\ j_2\to C$ and hence be considered regular.

Towards the resolution of the above issue, we may want to imagine a modified version of hybrid approach where pathway decomposition and bisimulation are not merely composed as in the above definition, but combined in a more integrated manner. For example, we have considered a kind of ``integrated''  hybrid approach in which regularity and the delimiting condition of weak bisimulation \cite{Qing,Johnson2016} are tested only after we apply the interpretation $m$ to the prime pathways in the elementary basis. While empirical results suggest that such modifications may successfully fix the issue described above, their theoretical implications are yet to be understood.


\section{Case studies}\label{sec:case}
In this section, we study five real-life examples from the field of DNA computing in order to demonstrate how the theory developed in this paper can be applied in practice. The code that was used to test these examples is included as part of the \emph{Nuskell} suite for compiling and verifying DNA implementations (previously called BioCRN \cite{S11}). Nuskell interfaces the formal basis enumeration algorithm from Section \ref{sec:algorithm} with other software pieces to form the following pipeline for verifying CRN implementations. First, the given target CRN is converted into a set of DNA molecules using the Nuskell compiler \cite{S11}. Second, all the reactions that can occur between these DNA molecules are enumerated using Grun et al.'s domain-level DNA reaction enumerator \cite{Grun}, from which the fuel species are pruned out as described in Section \ref{section:fuel}. The resulting reactions constitute the implementation CRN. Finally, either pathway decomposition, weak bisimulation \cite{Qing,Johnson2016}, or the compositional hybrid approach can be applied to the target CRN and the implementation CRN to verify that the two are indeed equivalent.

The current version of Nuskell implements a special case of the compositional hybrid approach in which $\mathcal V=\mathcal X$ (see Definition \ref{def:hybrid}). In other words, the Nuskell compiler provides the verifier with not only the target and implementation CRNs, but also information on formal and waste labeling. For formal labeling, it uses a pattern matching algorithm described in \cite{S11} to decide which species in the implementation CRN should correspond to which formal species. For waste labeling, it currently uses the following criterion proposed in \cite{S11}:
\begin{definition}
A species is a \textbf{non-waste} if it is formal or it is a reactant of a reaction that involves at least one non-waste either as a reactant or a product.
An intermediate species that is not a non-waste is a \textbf{waste} species.
\label{def:waste}
\end{definition}
Given this information, our verifier first finds the formal basis of the implementation CRN with respect to exactly those species labeled as formal or waste species by the Nuskell compiler, and then verifies that this formal basis is weak bisimulation equivalent to the target CRN under the natural interpretation that accords with the compiler's labeling.

\subsection*{Example \#1}

For the first example, we will implement the target CRN $\{A\to X+Y+Z,\ A\to X+Y,\ A\to X,\ A\to B\}$ using the translation scheme proposed in \cite{SSW10}. Figure~\ref{fig:soloveichik} shows how a unimolecular reaction (i.e.~reaction with exactly one reactant) gets implemented in this translation scheme.

\begin{figure}
\centering
\includegraphics[width=\textwidth]{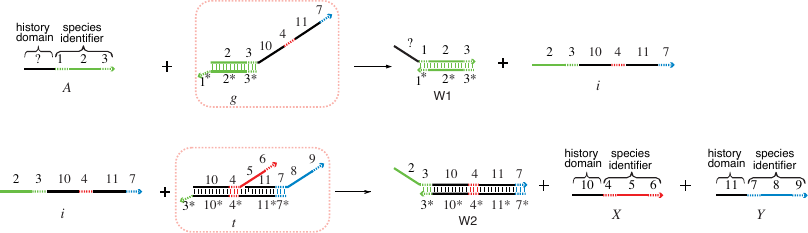}
\caption {The implementation of $A\to X+Y$ according to the translation scheme from \cite{SSW10}. Dotted boxes indicate fuel species. Note that in this translation scheme each formal species molecule will retain a ``history domain'' specific to the gate from which it was produced.
}\label{fig:soloveichik}
\end{figure}

To understand how this translation works, first note that the figure makes use of domain-level annotation as opposed to sequence-level annotation: that is, the DNA strands in Figure \ref{fig:soloveichik} are specified by numbered segments, or ``domains,'' instead of the actual nucleotide base sequences of A, G, C, and T. Here, the star is used to indicate sequence complementarity, e.g.~segments 1 and 1* are complementary to each other. Under the assumption that domains otherwise have very little complementarity, this abstraction is very useful in modeling complex DNA systems.

In this implementation, $g$ and $t$ are fuel species that are assumed to be present in large concentration. Hence, when the molecule $A$ is present in the solution, $A$ and $g$ may collide and bind to each other by the 1 and 1* domains that are complementary to each other. When this happens, since the adjacent domains 2 and 3 on $A$ and 2* and 3* on the bottom strand of $g$ are also complementary to each other, the hybridization can continue to the right, by a process called branch migration, thus displacing the top strand of $g$ and producing two species on the right-hand side of the first reaction. The resulting molecule $i$ can then react with another fuel species $t$ to produce the desired products $X$ and $Y$. Note that $g$ and $t$ can be easily modified to implement reactions with different numbers of product molecules, e.g.~$A\to X$ or $A\to X+Y+Z$. 

Note also that this scheme makes use of ``history domains'' in implementing formal species, represented by the question mark in the domain specification. For instance, in this example, any single-stranded molecule that has an arbitrary domain followed by domains 1, 2, and 3 is considered to be  $A$. This is necessary because if the implementation consists of multiple modules like the one depicted in this figure, $A$ molecules produced by different modules will have different history domains, each specific to the gate from which the molecule was produced. 
However, all those different versions of $A$ will then be able to participate in the same set of downstream modules, because as can be seen in the figure, the history domains do not participate in the reactions employed by those modules.

If we follow this translation scheme blindly, we would require exactly two fuel species for each unimolecular reaction in the target CRN. However, in the case of our target CRN $\{A\to X+Y+Z,\ A\to X+Y,\ A\to X,\ A\to B\}$, there is an optimization technique we can use to reduce the number of fuel species, exploiting the fact that the first three reactions in this target CRN are very similar to one another. Namely, it turns out that in this case we can share one fuel species $g$ among those three reactions, therefore using only $6$ fuel species to implement the four reactions in the target CRN rather than $2\times4=8$. In practice, researchers who experiment with actual DNA systems generally want to employ such optimizations whenever possible, because they are often crucial to the cost and efficacy of the experiment. Figure \ref{fig:casestudy1} illustrates how the optimized implementation works for this example.

\begin{figure}
\centering
\includegraphics[width=\textwidth]{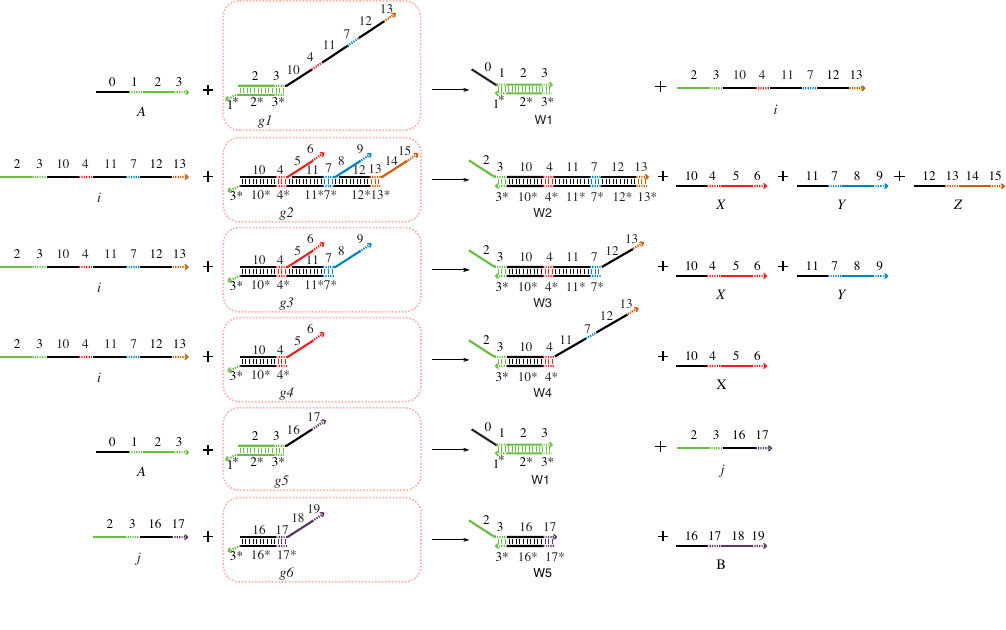}
\caption {An optimized implementation of target CRN $\{A\to X+Y+Z,\ A\to X+Y,\ A\to X,\ A\to B\}$.}\label{fig:casestudy1}
\end{figure}

In the next step, we preprocess the fuel and waste species to obtain a simpler CRN that involves only formal and intermediate species. Although in principle waste species are handled using the compositional hybrid approach, we will assume for the sake of presentation that in this example we can treat waste species in the same way as fuel species. We note that this assumption is not far-fetched because in this example it is rather obvious that the waste species do not participate in any reaction at all. After this preprocessing, the resulting implementation CRN looks as follows:
\begin{align*}
A&\to i\\
i&\to X+Y+Z\\
i&\to X+Y\\
i&\to X\\
A&\to j\\
j&\to B
\end{align*}
Moreover, the optimized implementation no longer uses multiple history domains for a single formal species, so the theory of pathway decomposition can be directly applied to this implementation.

Before we proceed, we also remark that this example contains a notable instance of the delayed choice phenomenon, and therefore cannot be verified by either weak bisimulation \cite{Qing,Johnson2016} nor serializability \cite{Lakin}. Namely, the intermediate $i$ has multiple fates $\llbrace X,Y,Z\rrbrace$, $\llbrace X,Y\rrbrace$, and $\llbrace X\rrbrace$, and hence it is unclear what its instantaneous interpretation should be. Indeed, we cannot interpret $i$ to be any of $\llbrace X,Y,Z\rrbrace$, $\llbrace X,Y\rrbrace$, and $\llbrace X\rrbrace$ because then the CRN would appear to contain reactions $X+Y+Z\to X+Y$, $X+Y\to X+Y+Z$, and $X\to X+Y+Z$, respectively. Neither can we interpret $i$ to be $A$, because $i$ cannot turn into $B$.

In contrast, pathway decomposition has no difficulty verifying this implementation CRN. Running the algorithm for enumerating the basis, we find that its elementary basis is
\begin{align*}
\{&(A\to i,\ i\to X+Y+Z),\\
&(A\to i,\ i\to X+Y),\\
&(A\to i,\ i\to X),\\
&(A\to j,\ j\to B)\},
\end{align*}
from which it is clear that the CRN is tidy and regular, and moreover its formal basis equals the target CRN that we desired to implement. Hence, this implementation is pathway decomposition equivalent to the target CRN.

\subsection*{Example \#2}
Our second example is a verification of the network condensation procedure proposed in \cite{Grun}. The domain-level reaction enumerator from \cite{Grun} can produce the output using several different semantics. One is ``detailed'' semantics, in which all internal configuration changes within DNA molecules are enumerated step by step, one domain change at a time. Another is ``condensed'' semantics, in which internal configuration changes that occur within a single molecule are considered to be one step. For instance, we note that Figure \ref{fig:soloveichik} is an example of condensed semantics. If the first reaction in Figure \ref{fig:soloveichik} was enumerated using detailed semantics instead, it would be enumerated as three reactions instead of one, where the first reaction would be $A$ binding to $g$ by domain 1, the second reaction would be domain 2 on $A$ hybridizing with domain 2* on the bottom strand of $g$ (i.e.~domain 2* on the top strand of $g$ would now be displaced), and the third reaction would be domain 3 on $A$ hybridizing with domain 3* on the bottom strand of $g$ (i.e.~the top strand is now completely released). For practical purposes, it is often convenient to use condensed semantics, which produces many fewer species and reactions while still capturing the essential behavioral features of the given system.

While \cite{Grun} provides its own theoretical justification for the correctness of this condensation procedure, it would be interesting to verify using pathway decomposition that the CRNs generated by the two different semantics are indeed equivalent. For example, let us consider the DNA system in Figure \ref{fig:casestudy2}, enumerated using detailed and condensed semantics of Grun et al.'s enumerator \cite{Grun}. This example was taken from Figure 4 of \cite{Grun} and was slightly modified to highlight the delayed choice phenomenon inherent in the system. Condensed semantics identifies species that differ only by a reversible change of internal configuration with one another, resulting in ``resting sets'' of species variants that are easily interconvertible. In this example, species $G$ and $i4$ are grouped together and will be treated as one species $G$ in the condensed CRN (and similarly $D$ and $i7$ as one species $D$). In general, this grouping gives rise to subtle issues that necessitate the use of the compositional hybrid approach (see Example \#4), but we will show that the system at hand is simple enough that we can verify it using only pathway decomposition. To achieve this, we treat the CRN generated by condensed semantics as a target CRN and the CRN generated by detailed semantics as an implementation CRN.

\begin{figure}
\centering
\begin{subfigure}[t]{\textwidth}
  \centering
  \includegraphics[width=.85\textwidth]{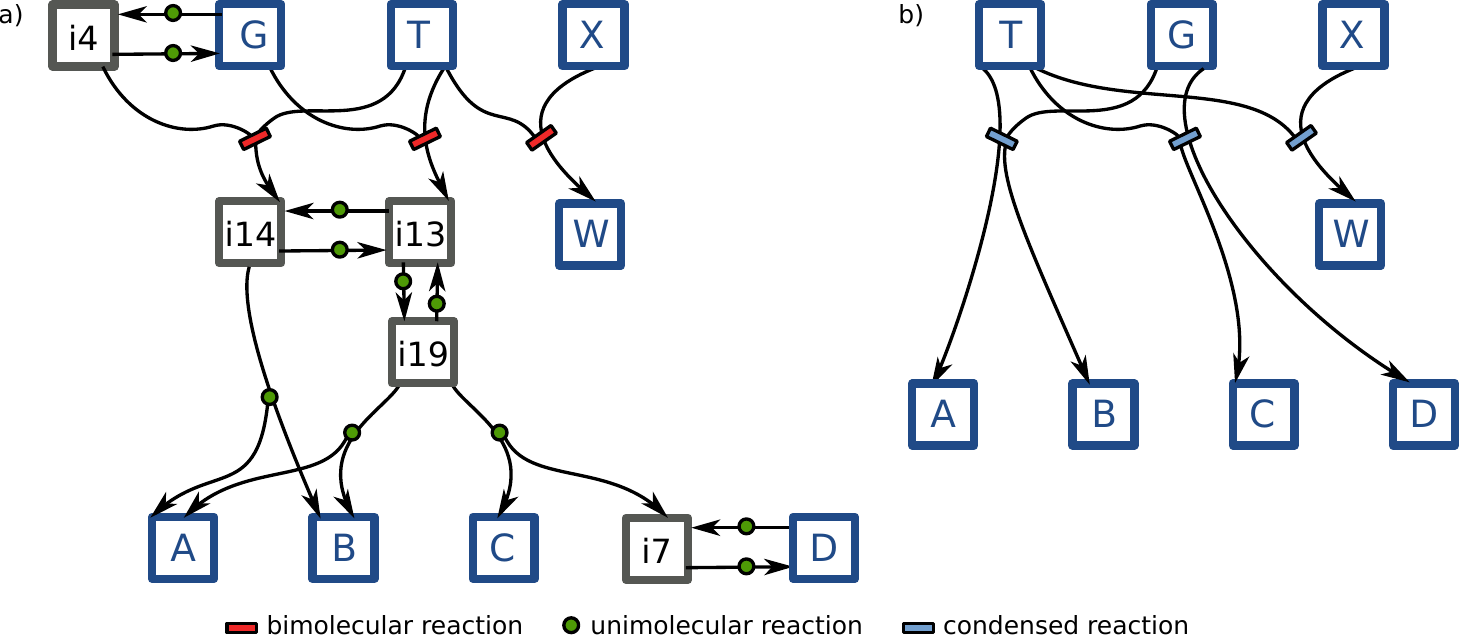}
\end{subfigure}
\vspace{.1in}
\begin{subfigure}[t]{\textwidth}
  \centering
  \includegraphics[width=.9\textwidth]{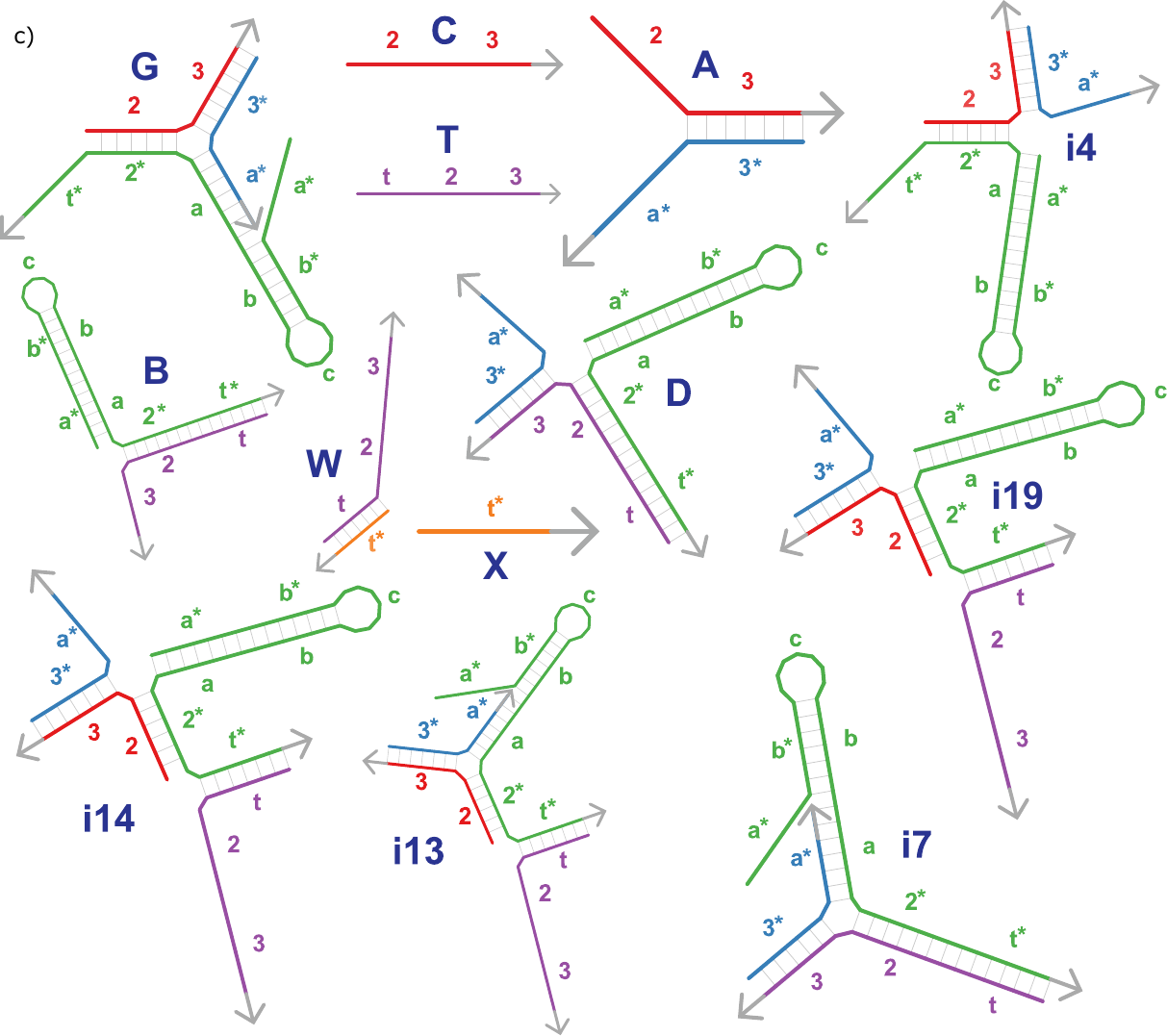}
\end{subfigure}
\caption {An example DNA system enumerated using detailed and condensed semantics. \textbf{a)} Detailed semantics enumeration. \textbf{b)} Condensed semantics enumeration. \textbf{c)} Key species shown in detail (adapted from Figure 4 of \cite{Grun}).}\label{fig:casestudy2}
\end{figure}

First of all, observe that the weak bisimulation approach of \cite{Qing,Johnson2016} is not sufficient to verify the correctness of this condensation. For example, the intermediate species $i19$ does not admit an appropriate instantaneous interpretation. If it is interpreted to be $\llbrace A,B\rrbrace$, the system would appear to contain the reaction $A+B\to C+D$. If it is interpreted to be $\llbrace C,D\rrbrace$, the system would appear to contain the reaction $C+D\to A+B$. It cannot be interpreted as $\llbrace G,T\rrbrace$, because it cannot react with an $X$ molecule to produce $W$.

In contrast, it is easily verified using our algorithm that the detailed CRN is tidy and regular and moreover its formal basis is $\{D\to D,\ G\to G,\ G+T\to A+B,\ G+T\to C+D,\ T+X\to W\}$, which is equal to the condensed CRN up to addition or removal of trivial reactions. Therefore the condensation in this example is correct according to pathway decomposition equivalence.

\subsection*{Example \#3}

As a third example, we consider the same translation scheme and optimization technique as in Example \#1, this time applied to a different target CRN $\{A+X\to X+X+A,\ A+X\to X+X,\ A+X\to X,\ A\to A+X+X,\ A\to A+X\}$. While this example is very similar in flavor to Example \#1, we can no longer directly verify it with pathway decomposition because now there can be multiple history domains for a single formal species. As can be seen in Figure \ref{fig:casestudy3}, there are four different implementation species that correspond to the formal species $X$ and two different implementation species that correspond to the formal species $A$. Moreover, as can be seen in Figure \ref{fig:casestudy32}, bisimulation still does not apply because of the delayed choice of species $j$. Hence, in this case the compositional hybrid approach is necessary to establish the equivalence between the target CRN and the implementation CRN.

\begin{figure}
\centering
\includegraphics[width=\textwidth]{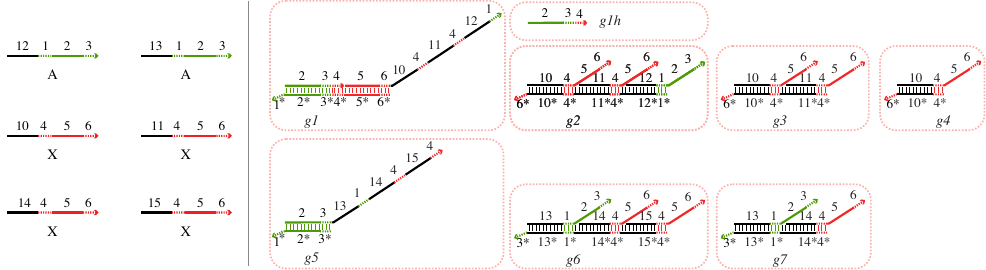}
\caption {An optimized implementation of target CRN $\{A+X\to X+X+A,\
  A+X\to X+X,\ A+X\to X,\ A\to A+X+X,\ A\to A+X\}$.  (Left) All
  species of $A$ and $X$ with distinct history domains. (Right) All fuel
  species of the optimized implementation.}\label{fig:casestudy3}
\end{figure}

To verify this implementation using the compositional hybrid approach, we first enumerate species and reactions using \cite{Grun}, then preprocess the fuel species to obtain the implementation CRN of Figure \ref{fig:casestudy32}, and finally run our formal basis enumerator on it treating $A_i$'s, $X_i$'s, and $W_i$'s as formal species. It verifies that the implementation CRN is tidy and regular and returns the formal basis shown in Figure \ref{fig:casestudy32}. While the formal basis turns out to be very large because of the existence of multiple history domains, it is easily shown to be weak bisimulation equivalent to the target CRN under the obvious interpretation $m(A_i)=\llbrace A\rrbrace, m(X_i)=\llbrace X\rrbrace, m(W_i)=\emptyset$. Therefore the given target CRN and implementation CRN are equivalent according to the compositional hybrid approach.

\begin{figure}
\centering
\begin{minipage}[b]{0.35\linewidth}
\centering
\begin{align*}
    A_1 &\to i_1\\
    i_1 &\to A_1\\
    A_2 &\to i_2\\
    i_2 &\to A_2\\
    i_1 + X_1 &\to j + W_3\\
    i_1 + X_2 &\to j + W_4\\
    i_1 + X_3 &\to j + W_5\\
    i_1 + X_4 &\to j + W_6\\
    i_2 + X_1 &\to j + W_7\\
    i_2 + X_2 &\to j + W_8\\
    i_2 + X_3 &\to j + W_9\\
    i_2 + X_4 &\to j + W_{10}\\
    j &\to X_3 + W_{11}\\
    j &\to X_3 + X_4 + W_{12}\\
    j &\to A_2 + X_3 + X_4 + W_{13}\\
    A_1 &\to k + W_1\\
    A_2 &\to k + W_2\\
    k &\to A_1 + X_1 + X_2 + W_{14}\\
    k &\to A_1 + X_1 + W_{15}
\end{align*}
\vspace{0.3cm}
\textbf{Implementation CRN}
\vspace{0.5cm}
\begin{align*}
    A &\to A\\
    A + X &\to X\\
    A + X &\to X+X\\
    A + X &\to A+X+X\\
    A &\to A+X+X\\
    A &\to A+X
\end{align*}
\vspace{0.3cm}
\textbf{After applying the interpretation}

\end{minipage}
\begin{minipage}[b]{0.6\linewidth}
\centering
\begin{align*}
    A_1 &\to A_1\\
    A_2 &\to A_2\\
    A_1 + X_1 &\to X_3 + W_{11} + W_3\\
    A_1 + X_2 &\to X_3 + W_{11} + W_4\\
    A_1 + X_3 &\to X_3 + W_{11} + W_5\\
    A_1 + X_4 &\to X_3 + W_{11} + W_6\\
    A_2 + X_1 &\to X_3 + W_{11} + W_7\\
    A_2 + X_2 &\to X_3 + W_{11} + W_8\\
    A_2 + X_3 &\to X_3 + W_{11} + W_9\\
    A_2 + X_4 &\to X_3 + W_{11} + W_{10}\\
    A_1 + X_1 &\to X_3 + X_4 + W_{12} + W_3\\
    A_1 + X_2 &\to X_3 + X_4 + W_{12} + W_4\\
    A_1 + X_3 &\to X_3 + X_4 + W_{12} + W_5\\
    A_1 + X_4 &\to X_3 + X_4 + W_{12} + W_6\\
    A_2 + X_1 &\to X_3 + X_4 + W_{12} + W_7\\
    A_2 + X_2 &\to X_3 + X_4 + W_{12} + W_8\\
    A_2 + X_3 &\to X_3 + X_4 + W_{12} + W_9\\
    A_2 + X_4 &\to X_3 + X_4 + W_{12} + W_{10}\\
    A_1 + X_1 &\to A_2 + X_3 + X_4 + W_{13} + W_3\\
    A_1 + X_2 &\to A_2 + X_3 + X_4 + W_{13} + W_4\\
    A_1 + X_3 &\to A_2 + X_3 + X_4 + W_{13} + W_5\\
    A_1 + X_4 &\to A_2 + X_3 + X_4 + W_{13} + W_6\\
    A_2 + X_1 &\to A_2 + X_3 + X_4 + W_{13} + W_7\\
    A_2 + X_2 &\to A_2 + X_3 + X_4 + W_{13} + W_8\\
    A_2 + X_3 &\to A_2 + X_3 + X_4 + W_{13} + W_9\\
    A_2 + X_4 &\to A_2 + X_3 + X_4 + W_{13} + W_{10}\\
    A_1 &\to A_1 + X_1 + X_2 + W_{14} + W_1\\
    A_1 &\to A_1 + X_1 + W_{15} + W_1\\
    A_2 &\to A_1 + X_1 + X_2 + W_{14} + W_2\\
    A_2 &\to A_1 + X_1 + W_{15} + W_2
\end{align*}
\vspace{0.3cm}
\textbf{Formal basis}
\end{minipage}
\caption {Applying the compositional hybrid approach to Example \#3.}\label{fig:casestudy32}
\end{figure}

We remark that even though this implementation CRN contained 25 species and 19 reactions, the basis enumeration algorithm took less than one second to finish on an off-the-shelf laptop computer. This performance is perhaps surprising considering that the algorithm is based on brute-force enumeration, and it suggests that despite the adverse worst-case time complexity, the algorithm may still be practical for many instances that arise in practice. In fact, when we made full use of the optimization techniques outlined in Sections \ref{sec:modular} and \ref{sec:optimization}, the algorithm terminated on almost all of our test instances in less than ten seconds. The performance was particularly strong for implementations that had monomolecular substructure. For example, the translation scheme from \cite{SSW10} applied on a target CRN consisting of 20 reactions produces an implementation CRN consisting of 394 reactions and 387 species. However, since this implementation CRN is modular and has monomolecular substructure, our verifying algorithm was able to successfully verify it in mere 4 seconds. On the other hand, we report that the algorithm failed to terminate in an hour on some instances that did not have monomolecular substructure.


\subsection*{Example \#4}
In our fourth example, we illustrate why pathway decomposition may not
suffice for the verification of the condensation procedure in \cite{Grun} and how the compositional hybrid approach can be used to remedy this problem. This example is almost identical to the system in Example \#2, except that we remove species $X$ and $W$ from the system and add a different species $Y$, which is merely the DNA strand consisting of a single domain `a'.

\begin{figure}
\centering
\begin{subfigure}[t]{0.45\textwidth}
  \includegraphics[width=.5\textwidth]{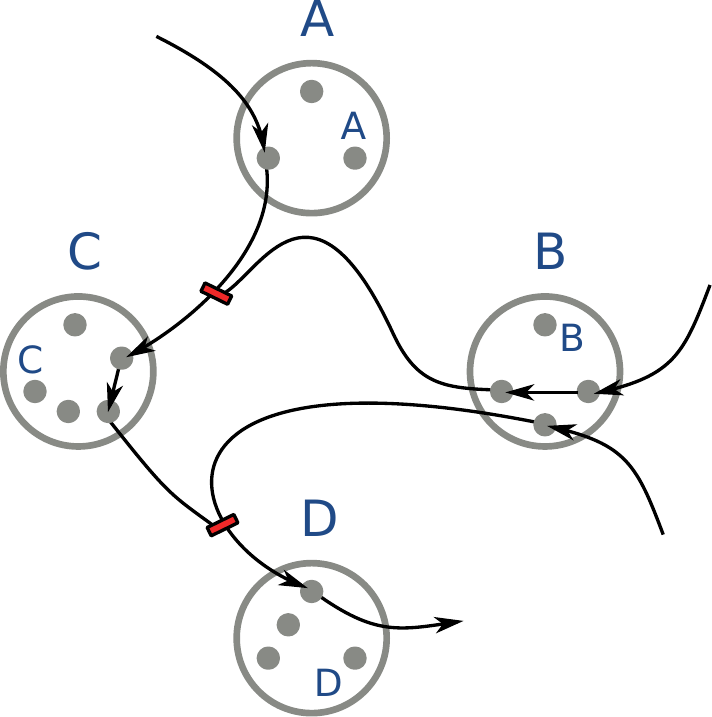}
  \begin{minipage}[b]{0.4\linewidth}
    \centering
      $? \to A$\\
      $? \to B$\\
      $? \to B$\\
      $B + A \to C$\\
      $C + B \to D$\\
      $D \to\ ?$\\~
  \end{minipage}
  \caption{}
\end{subfigure}
\qquad\qquad
\begin{subfigure}[t]{0.4\textwidth}
  \centering
  \includegraphics[width=.65\textwidth]{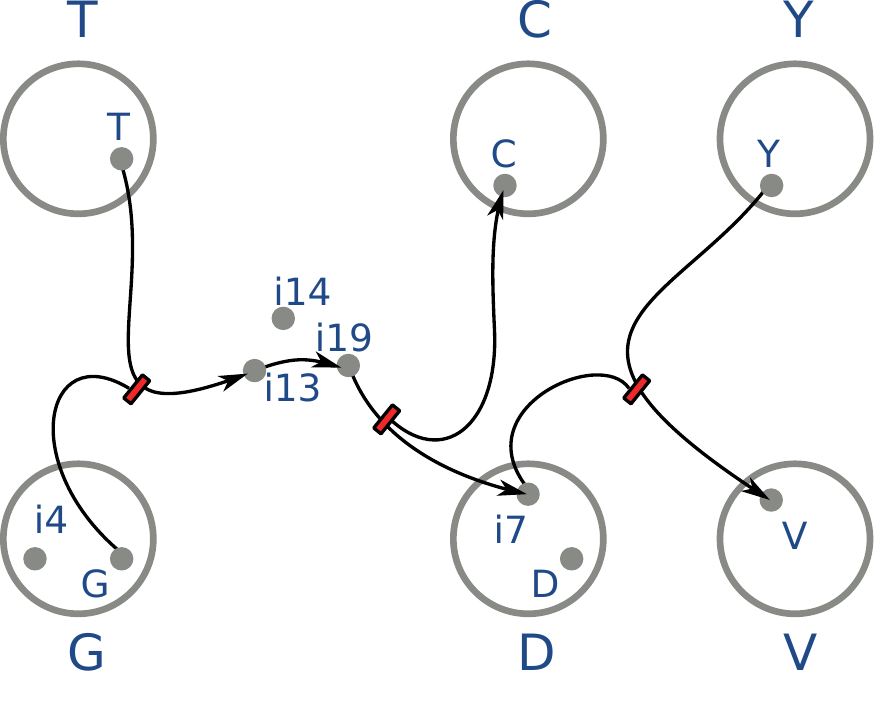}
  \caption{}
\end{subfigure}
\caption {Conceptual issues for pathways involving resting sets with multiple interconvertible species.   Species are illustrated as dots, while resting sets are indicated by circles encoding the interconvertible species. Although unimolecular reactions between the species within a given resting set form a strongly connected graph, for clarity they are not explicitly shown. (a) A graph of the resting sets and potential reaction
  pathways for the adjacent implementation CRN.  Reaction pathways do not necessarily involve formal species, but can instead utilize equivalent species from the same resting set. (b) A
  subgraph of the resting sets and reaction pathways for the CRN of
  Example \#4 that illustrates the pathway
  $G+T\to i13,\ i13\to i19,\ i19\to C+i7,\ i7+Y\to
  V$.}\label{fig:ten-minus-epsilon}
\end{figure}

Since there are many species in this system that have an unhybridized a* domain (e.g. $G$, $i4$, $A$, $i7$, and $D$), molecule $Y$ can react with those species to form various other species, thus giving rise to a reaction network that is much more complex than in Figure \ref{fig:casestudy2}. Most importantly, we note that molecule $i7$ can also bind with $Y$ to form a species which we shall call $V$. Since $i7$ is identified with $D$ in condensed semantics, this reaction will appear as $D+Y\to V$ in the condensed CRN. However, if we try to run pathway decomposition on the detailed CRN treating as formal species only those species that appear in the condensed CRN, we will find that the prime pathway $G+T\to i13,\ i13\to i19,\ i19\to C+i7,\ i7+Y\to V$ is irregular and thus pathway decomposition does not apply to this system (see Figure~\ref{fig:ten-minus-epsilon}). The problem here is that species $i7$ should have been considered a formal species because it is identified with $D$, even though its name does not appear in the condensed CRN. Hence, in order to apply pathway decomposition properly, we would need a way to inform the theory that $i7$ is also a copy of the species $D$.

\begin{figure}
\centering
\begin{minipage}[b]{0.6\linewidth}
\centering
\begin{align*}
    i13 &\leftrightharpoons i14\\
    i13 &\leftrightharpoons i19\\
    i14 &\to A + B\\
    i38 &\leftrightharpoons i40\\
    i38 &\leftrightharpoons i46\\
    i40 &\to B + Z\\
    i46 &\to B + Z\\
    i46 &\to C + V\\
    i19 &\to A + B\\
    i19 &\to C + i7\\
    A + Y &\to Z\\
    D &\leftrightharpoons i7\\
    D + Y &\to i41\\
    i7 + Y &\to V\\
    G &\leftrightharpoons i4\\
    G + T &\to i13\\
    G + Y &\to U\\
    i4 + T &\to i14\\
    i4 + Y &\to i42\\
    T + U &\to i38\\
    T + i42 &\to i40\\
    U &\leftrightharpoons i42\\
    V &\leftrightharpoons i41
\end{align*}
\vspace{0.3cm}
\textbf{Implementation CRN}
\vspace{0.2cm}
\begin{align*}
    m(A)=\llbrace A\rrbrace,&\quad\quad m(B)=\llbrace B\rrbrace,\\
    m(C)=\llbrace C\rrbrace,&\quad\quad m(D)=m(i7)=\llbrace D\rrbrace,\\
    m(G)=m(i4)=\llbrace G\rrbrace,&\quad\quad m(T)=\llbrace T\rrbrace,\\
    m(i42)=m(U)=\llbrace U\rrbrace,&\quad\quad m(Y)=\llbrace Y\rrbrace,\\
    m(i41)=m(V)=\llbrace V\rrbrace,&\quad\quad m(Z)=\llbrace Z\rrbrace
\end{align*}
\vspace{0.3cm}
\textbf{Bisimulation interpretation $m$}
\end{minipage}
\begin{minipage}[b]{0.35\linewidth}
\centering
\begin{align*}
    A + Y &\to Z\\
    D &\to i7\\
    D + Y &\to i41\\
    i7 &\to D\\
    i7 + Y &\to V\\
    G &\to i4\\
    G + T &\to A + B\\
    G + T &\to C + i7\\
    G + Y &\to U\\
    i4 &\to G\\
    i4 + T &\to A + B\\
    i4 + T &\to C + i7\\
    i4 + Y &\to i42\\
    T + U &\to B + Z\\
    T + U &\to C + V\\
    T + i42 &\to B + Z\\
    T + i42 &\to C + V\\
    U &\to i42\\
    i42 &\to U\\
    i41 &\to V\\
    V &\to i41
\end{align*}
\vspace{0.3cm}
\textbf{Formal basis}
\vspace{0.1cm}
\begin{align*}
    A + Y &\to Z\\
    D + Y &\to V\\
    G + T &\to A + B\\
    G + T &\to C + D\\
    G + Y &\to U\\
    T + U &\to B + Z\\
    T + U &\to C + V
\end{align*}
\vspace{0.3cm}
\textbf{Condensed CRN}
\end{minipage}
\caption {Applying the compositional hybrid approach to Example \#4.}\label{fig:casestudy4}
\end{figure}

We note that this problem is very similar to the problem of multiple history domains that we discussed in Example \#3. Just as $A_1$ and $A_2$ had both to be considered an implementation of $A$, in this example we need to ensure that both $D$ and $i7$ are considered an implementation of $D$. Thus, we can apply the compositional hybrid approach to such systems by first running pathway decomposition on the detailed CRN, treating as formal species all species from the condensed CRN and any other species that are identified with those species within resting sets, and then verifying that the resulting formal basis is weak bisimulation equivalent to the condensed CRN. In our example, this would correspond to treating species like $i4$ and $i7$ as formal species, and then using the interpretation $m(G)=m(i4)=\llbrace G\rrbrace, m(D)=m(i7)=\llbrace D\rrbrace$ in the bisimulation step of the compositional hybrid approach. All of this information, i.e.~which species should be considered formal and what interpretation should be used, is provided to our verifier software by the reaction enumerator of \cite{Grun}. This way, the irregular pathway from the previous paragraph would no longer be prime, because it can now be decomposed into $G+T\to i13,\ i13\to i19,\ i19\to C+i7$ and $i7+Y\to V$. The full result of a compositional hybrid approach verification of this example is shown in Figure \ref{fig:casestudy4}, where we can easily check that the formal basis of the detailed CRN is weak bisimulation equivalent to the condensed CRN under the interpretation $m$.

Like Example \#3, we note that this is an example of a DNA system that neither pathway decomposition nor bisimulation can verify, but the compositional hybrid approach can.

\subsection*{Example \#5}
For the final example, we investigate the following CRN as an implementation of $\{A+B\to C+D+E\}$:
\begin{align*}
A&\leftrightharpoons i+j\\
i+B&\leftrightharpoons k+l\\
k&\to C\\
l&\leftrightharpoons m+n\\
m&\to D\\
j+n&\to E
\end{align*}
Unlike the previous examples, this CRN was not constructed by a direct application of a published CRN implementation scheme, although it was inspired by the reaction network given rise to by the ``garbage collection'' module of \cite{C09}. Nonetheless, this system is very interesting because it distinguishes pathway decomposition from bisimulation \cite{Qing,Johnson2016} or serializability \cite{Lakin} without making use of the delayed choice phenomenon. Therefore it suggests that delayed choice may in fact be just one example of many interesting behaviors that are allowed by pathway decomposition but not by other approaches.

First, we observe that this implementation would be deemed ``incorrect'' by bisimulation or serializability. To see this for bisimulation, we simply note that the interpretation of $k$ must be $C$ because of reaction $k\to C$. However, this would result in $k+l\to i+B$ being interpreted as some reaction that consumes at least one $C$ and produces at least one $B$, which means that this CRN cannot be a correct implementation of $\{A+B\to C+D+E\}$ according to bisimulation. To see it for serializability, we note that, roughly translated into our language, serializability requires every prime pathway in a module to have a well-defined turning point reaction and moreover visit the same set of states prior to the turning point reaction. In this example, the turning point reaction of a prime pathway would be defined as the first irreversible reaction that occurs in the pathway, e.g.~reaction $k\to C$ in pathway $(A\to i+j,\ i+B\to k+l,\ k\to C,\ l\to m+n,\ m\to D,\ j+n\to E)$ and reaction $j+n\to E$ in pathway $(A\to i+j,\ i+B\to k+l,\ l\to m+n,\ j+n\to E,\ k\to C,\ m\to D)$.\footnote{Lakin et al.'s definition of a turning point (which they call a ``commit reaction'') is slightly different from our definition. For detail, refer to \cite{Lakin}.} Noting that these two pathways do not visit the same set of states prior to the respective turning point reactions, we conclude that this implementation cannot be handled by serializability. However, we do remark that there could be a relaxed version of serializability which can handle this example.

In contrast to the above two approaches, pathway decomposition does not have any difficulty with this example. Not only does it allow prime pathways implementing the same reaction to have different choices of turning point reactions, but it also allows them to go through different sets of states as long as each of those prime pathways is regular. In fact, our basis enumeration algorithm easily finds the formal basis of the given implementation CRN to be $\{A\to A,\ A+B\to A+B,\ A+B\to C+D+E\}$, proving that it is a correct implementation of $\{A+B\to C+D+E\}$ according to pathway decomposition.

\section{Conclusions}

The development of pathway decomposition theory was motivated by a desire for a general notion of CRN behavioral equivalence, up to ignoring rate constants and implementation intermediates.  An overarching challenge is that, despite the set of species and set of reactions both being finite, the set of possible initial states may be infinite, the set of system states reachable from a given initial state may be infinite, and the set of possible pathways from a given initial state may be infinite -- yet we desire a guarantee that the available behaviors within two CRNs are essentially identical, and we desire that this guarantee may be found (or refuted) algorithmically in all cases.  These factors eliminate many standard approaches -- such as those that only handle finite state spaces -- from consideration. How well does the pathway decomposition approach meet these goals? 

The central concepts of pathway decomposition are quite general, allowing application of the theory to a wide range of CRNs, but some important limitations are imposed.  First, species must be divided into formal species and intermediate species, and we are only concerned with ``formal'' pathways of reactions that start with purely formal states and end in purely formal states.  (I.e. our theory does not concern itself with what may or may not happen when you start the CRN with intermediate species; they occur only in the middle of formal pathways.)  The basic idea is that any such formal pathway can be decomposed (perhaps not uniquely) into interleaved sub-pathways until non-decomposable (``prime'') pathways are reached.  The set of all such prime pathways defines the formal basis for the CRN -- the corresponding set of initial and final states for the set of prime pathways -- and two CRNs with the same formal basis are deemed equivalent.   The consequence is that the sets of formal pathways in the two CRNs can be put in correspondence with each other by decomposing into prime pathways and replacing each prime pathway by its formal basis reaction.  In this sense, anything that one CRN can do, can also be done by the other.  However, pathway decomposition theory applies only to CRNs that are tidy (any state reached from a formal state can clean up all intermediate species and return to a formal state) and regular (every prime pathway consists of a consumptive phase followed by a productive phase, separated by a turning point reaction).  

The choices implicit in the formulation of pathway decomposition allow for a general and elegant theory.  A primary feature is that other than regularity, there are no structural constraints on what goes on inside the prime pathways.  This provides the potential for intermediates within a pathway to perform a non-trivial (deterministic or non-deterministic) computation.  In particular, intermediates can be shared between prime pathways, and intermediates may not ``know'' which pathway they are on -- a phenomenon we call ``delayed choice''.   This is substantially less restrictive than other related methods for CRN implementation verification \cite{Qing,Johnson2016,Lakin}.  On the other hand, in the case that a CRN can be divided into two parts with distinct intermediates, a modularity property holds such that the formal bases of the two parts can be considered independently (as in a previous method \cite{Lakin}). This greatly facilitates algorithmic verification of CRN implementations.   Finally, because the formal basis of an implementation CRN is unique (up to choice of which species are considered formal and which are intermediates), a given implementation CRN cannot be considered a correct implementation of two distinct formal CRNs -- a natural property that, again, does not hold for some related methods \cite{Qing,Johnson2016}.

On the other hand, the elegance and generality of pathway decomposition theory come at a cost.  Because the core theory only addresses tidy CRNs, fuel and waste species must be removed by pre-processing outside the core pathway decomposition theory. 
Because fuels are presumed to be held at constant concentrations, they can be eliminated from the CRN representation with a change of rate constants and absolutely no effect on the dynamics.  However, implementation species considered ``waste'' are often not entirely inert, but rather their interactions with the system are such that their presence or absence does not affect the possible formal pathways -- for example, they might interact only with other ``waste'' species.
Further, because the core theory requires a one-to-one correspondence between formal species and selected representative species in the implementation, the core theory is insufficient for implementation schemes where a formal species may be represented by molecules with variable regions, such as the ``history domains'' of Soloveichik et al \cite{SSW10}.  To accommodate these concerns, we developed a compositional hybrid theory, which allows pathway decomposition to be first applied to the implementation CRN with all waste species and all representations of formal species being designated as ``formal'', after which bisimulation \cite{Qing,Johnson2016} is applied to the resulting formal basis to establish correctness with respect to the original formal CRN.   For both the core theory and the compositional hybrid theory, ``correctness'' of an implementation provides guarantees that every pathway in the formal CRN can occur in the implementation and vice versa, under an appropriate interpretation of implementation pathways.

The compositional hybrid theory is conceptually sufficient for verifying {-- or finding errors in --} CRNs implemented according to most published translation schemes\footnote{{We have applied pathway decomposition to translation schemes from \cite{SSW10, C09, C11, QSW11, chen2013programmable, Lakin}, verifying implementations using many schemes, revealing errors and suggesting fixes for some schemes, identifying concerns such a potential leak pathways in other schemes, and encountering limitations of the theory for still other cases.  These investigations will be reported in more detail elsewhere; some general observations are mentioned below.}}.
%
However, the algorithmic challenges of finding the formal basis for pathway decomposition theory, and of finding the interpretation function for bisimulation theory, are substantial:  the verification of correctness pertains to all possible initial states of the CRNs (which are infinite in number) and all possible pathways from these states (also infinite).  To meet this challenge, we developed the notion of ``signatures'' for partial pathways, proving that they are bounded in number for implementations for which the prime pathways have bounded width, and thus proving that at least in this case our algorithm is guaranteed to terminate, even if there are an infinite number of distinct prime pathways.  Although the worst-case complexity of our algorithm is unknown, in practice implemented CRNs have a modularity property that is easily recognized and exploited, often allowing verification to complete in a matter of seconds for systems of the scale that is currently experimentally feasible\footnote{{For example, a system of ten up-to-bimolecular reactions compiles, according to the translation scheme in \cite{SSW10}, into a implementation with 185 species and 150 reactions and verifies in about ten seconds on a 2013 MacBook Pro.}}.  
{Although verifying CRN bisimulation equivalence in the general case is PSPACE-complete \cite{Johnson2016}, the bisimulation test implemented for the compositional hybrid theory in Nuskell is a restricted case, and}
in practice this has not proven to be the limiting step for difficult verification cases.

During the course of our investigations, we encountered a number of CRN implementations that intuitively seem correct, but which are not accepted by our theory, pointing to the need for a yet more general notion of correctness.  We give {four} examples here.
The first was previously mentioned:  our notions of regularity and turning points, which appear necessary for correctly implementing irreversible reactions, preclude the use of physically reversible implementations of logically reversible reactions, such as $\{ A+B \rightleftharpoons i, i \rightleftharpoons j, j \rightleftharpoons C+D \}$.  
{Such implementations appear in  \cite{QSW11} and can reduce energy consumption and reduce implementation complexity exponentially \cite{thachuk2012space,condon2012less}.} 
Interestingly, the serializability approach to verification \cite{Lakin} shares this restriction, but the bisimulation approach \cite{Qing,Johnson2016} easily accommodates physically reversible implementations. Technically speaking, the compositional hybrid theory, which generalizes both pathway decomposition and bisimulation, can handle the above example, but it is only by ``abusing'' the theory by setting all species to be ``formal'' for its pathway decomposition step.
As a second example, one might have a reaction implementation where a ``waste'' species is produced prior to the turning point, e.g., 
$\{A  \to i + W, i  \to A, B + i \to C \}$.  The compositional hybrid theory will not accept this implementation because in the initial pathway decomposition step, where waste $W$ is temporarily considered ``formal'', the prime pathway $( A \to i + W, B+i \to C)$ is not a regular implementation of $A + B \to C + W$ because $W$ is produced prior to the arrival of $B$. 
{Although this type of situation does not commonly arise in DNA strand displacement systems, it could be ameliorated by simply eliminating waste species (as defined in Definition \ref{def:waste}) from the implementation CRN prior to verification.}
{A more general solution would be} to propose an ``integrated hybrid theory'' where regularity is tested only after the bisimulation interpretation has been applied; however, this complicates attempts to prove a theorem relating formal and implementation pathways analogous to Theorem~\ref{thm:hybrid}. 
The third example  concerns situations where formal species interact in non-meaningful ways, resulting in pathways that have no net effect.  As an illustration, in the implementation $\{A \rightleftharpoons i, i+B \to C, i + D \rightleftharpoons j \}$, the pathway $(A \to i, i+D \to j, j \to D+i, i+B \to C)$ is a prime pathway taking $A+B+D$ to $D+C$, but it is not regular.  Pathways like $(i+D \to j, j \to i+D)$, which end where they begin and which never produce a net formal species that wasn't previously consumed, could be called {\it futile loops}; they would arise, for example, when the implementation model explicitly accounts for fleeting binding between molecules that results in no transformation.    In some published CRN-to-DNA translation schemes, especially those that involve 
variable ``history'' domains and multiple reversible steps prior to the turning point {(e.g. \cite{C09}, and \cite{SSW10} generalized to trimolecular reactions)}, the involvement of futile loops can result in {irregular prime pathways that interlink multiple history-distinct versions of the same signal, 
such as $(A \to i, i+B_1 \to j_1, j_1 \to B_1+i, i+B_2 \to j_2, j_2 \to C)$}.
{Interestingly, the integrated hybrid theory could also address this particular problem, if it could be rigorously justified.}
{Alternatively, it} would be desirable to allow decomposition by removing such futile loops (in which case the above  pathway would be considered to implement $A+B \to C$, a sensible result) but unfortunately modifying our definitions to allow this leads to additional complexities for the notion of a ``signature'', and we have not been able to generalize our algorithm while retaining a proof that it is guaranteed to terminate.
{The fourth example highlights a situation where our algorithm's combinatorial explosion makes verification infeasible.  Specifically, most cases where verification is fast and scalable involve translation schemes that result in implementations that are modular and have monomolecular substructure, as per Sections \ref{sec:modular} and \ref{sec:optimization}.  However, the ``garbage collection'' stages in the translation schemes of \cite{C09} do not have monomolecular substructure, and our algorithms are capable of verifying only the simplest instances.  Interestingly, if the garbage collection stages are removed or key garbage collection species are held at constant concentrations as fuel, then monomolecular substructure is restored and verification proceeds apace.}
{In summary, there is} room for a deeper understanding of the notion of logical correctness for CRN implementations, and of the relative capabilities of different existing theories, such as composition and modularity properties.

Even in their present form, existing theories and algorithms for establishing the correctness of CRN implementations can play an important role in the development of rigorous compilers for molecular programming.   We envision that future compilers for molecular programming will conform to the standards established in electrical and computer engineering: complexity is managed using an abstraction hierarchy that allows a program specification in a high-level language (such as Verilog) to be translated through a series of intermediate-level languages to a physically implementable low-level language (such as transistor-level netlists). Furthermore, the language at each level of the hierarchy has well-defined semantics (the mathematical model describing the behavior of the program) and most importantly, formal proofs can establish that the compiler's transformation from a higher-level language to a lower-level language preserves the essential behavioral invariants~{\cite{leroy2009,Rival04}}.  For molecular programming with dynamic DNA nanotechnology, formal CRNs could serve as a higher-level language, while domain-level DNA strand displacement models could serve as a lower-level language -- the semantics of which provide the implementation CRN.  Note that in contrast to traditional compilers that (ideally) may be proved to be correct for all source programs, our pathway decomposition theory and algorithms are best suited for evaluating the correctness on a case-by-case basis.  While proving the general correctness of a CRN-to-DNA translation scheme would obviate the need for time-consuming verification algorithms to be run, the advantages of case-by-case verification are (1) when new translation schemes are proposed, they may immediately be used with confidence prior to establishing what may be a difficult general-case proof; (2) when a new low-level semantics is considered (e.g. either more or less detail in the molecular model), again there is no need to attempt a new proof for the general case; (3) in cases where a translation scheme is in fact not correct in the general case, it may still be used with confidence for CRNs that it does implement correctly; and (4) since it is highly desirable to make experimental systems as simple as possible, formal verification can be used to establish or refute the correctness of arbitrary attempts to optimize and simplify the DNA-level design.  These ideas have been implemented in the verifying compiler, Nuskell, which has already been used to catch bugs in several translation schemes~\cite{S11}.

Although the task of building an abstraction hierarchy for molecular programming with dynamic DNA nanotechnology seems particularly tractable, in principle the same formalism could be used to establish the correctness of other types of molecular and biochemical systems.  For example, systems of protein enzymes and nucleic acid substrates have been used to construct cell-free biochemical circuits~\cite{Kim2006,Montagne2011}; given a specification for the desired behavior as a formal CRN together with a CRN describing the actual implementation details, one could use pathway decomposition theory to examine its logical correctness.  Similarly, pathway decomposition could provide an alternative perspective on the validity of descriptions at multiple levels of details for biological networks studied in systems biology, or as an evaluation of coarse-graining and model reduction techniques~\cite{Ackermann2012}.  However, there are several limitations to pathway decomposition theory for these purposes.  First, in many such cases rate constants are important, but pathway decomposition theory does not consider them.  Related, in many cases approximate implementations are sufficient -- for example, if non-conforming pathways occur very rarely in the discrete stochastic (Gillespie) dynamics for the CRN, or if non-conforming pathways ``average out'' in the continuous deterministic (ordinary differential equation) dynamics~\cite{Klavins1, Klavins2, C14}.  
Finally, for some purposes the target behavior that the implementation aims to achieve is not best described as a CRN, but rather by some other specification language such as temporal logic.  Many of these issues are explored in the literature on Petri nets \cite{Heiner2008}.  
{However, a more fundamental concern is that CRNs are not an efficient representation for describing combinatorial processes in biology, for which more effective models have been developed~\cite{BNGL,danos2007} and analyzed~\cite{behr2016}.  It is reasonable to presume that such models could in the future provide a programming language for more sophisticated molecular machines.}
Nonetheless, the notion of logical correctness of CRN implementations that is provided by pathway decomposition theory has already proved its effectiveness for catching logical errors in CRN-to-DNA translation schemes and appears to be particularly suitable for incorporation into automated verifying compilers for molecular programming.

\section*{Acknowledgments}
We appreciate helpful discussions with John Baez, Luca Cardelli, 
Vincent Danos, Qing Dong, Robert Johnson, {Stefan Badelt,} and Matthew Lakin.
S.W.S. was supported by California Institute of Technology's Summer Undergraduate Research Fellowship 2009,
NSF grant CCF-0832824, 
ARO Grant W911NF-09-1-0440 and NSF Grant CCF-0905626.
C.T. was supported by NSF grants CCF-1213127, SHF-1718938, and a Banting Fellowship. 
E.W. was supported by NSF grants CCF-0832824, CCF-1213127, and CCF-1317694.

\bibliographystyle{plain}
\bibliography{pd}

\begin{thebibliography}{10}

\bibitem{Ackermann2012}
J\"org Ackermann, Jens Einloft, Joachim N\"othen, and Ina Koch.
\newblock Reduction techniques for network validation in systems biology.
\newblock {\em Journal of Theoretical Biology}, 315:71--80, 2012.

\bibitem{behr2016}
Nicolas Behr, Vincent Danos, and Ilias Garnier.
\newblock Stochastic mechanics of graph rewriting.
\newblock In {\em Proceedings of the 31st Annual ACM/IEEE Symposium on Logic in
  Computer Science}, pages 46--55. ACM, 2016.

\bibitem{BM2010}
Robert Brayton and Alan Mishchenko.
\newblock {ABC: An academic industrial-strength verification tool}.
\newblock In {\em Computer Aided Verification}, pages 24--40. Springer, 2010.

\bibitem{buchholz2008}
Peter Buchholz.
\newblock Bisimulation relations for weighted automata.
\newblock {\em Theoretical Computer Science}, 393:109--123, 2008.

\bibitem{C09}
Luca Cardelli.
\newblock Strand algebras for {DNA} computing.
\newblock {\em Natural Computing}, 10:407--428, 2011.

\bibitem{C11}
Luca Cardelli.
\newblock Two-domain {DNA} strand displacement.
\newblock {\em Mathematical Structures in Computer Science}, 23:247--271, 2013.

\bibitem{C14}
Luca Cardelli.
\newblock Morphisms of reaction networks that couple structure to function.
\newblock {\em BMC Systems Biology}, 8:84, 2014.

\bibitem{chen2013programmable}
Yuan-Jyue Chen, Neil Dalchau, Niranjan Srinivas, Andrew Phillips, Luca
  Cardelli, David Soloveichik, and Georg Seelig.
\newblock Programmable chemical controllers made from {DNA}.
\newblock {\em Nature nanotechnology}, 8:755--762, 2013.

\bibitem{condon2012less}
Anne Condon, Alan~J Hu, J{\'a}n Ma{\v{n}}uch, and Chris Thachuk.
\newblock Less haste, less waste: on recycling and its limits in strand
  displacement systems.
\newblock {\em Interface Focus}, 2:512--521, 2012.

\bibitem{danos2007}
Vincent Danos, J{\'e}r{\^o}me Feret, Walter Fontana, Russell Harmer, and Jean
  Krivine.
\newblock Rule-based modelling of cellular signalling.
\newblock In Lu\'is Caires and Vasco~T. Vasconcelos, editors, {\em
  {International Conference on Concurrency Theory (CONCUR 2007)}}, volume 4703
  of Lecture Notes in Computer Science, pages 17--41. Springer, 2007.

\bibitem{DK05}
Vincent Danos and Jean Krivine.
\newblock {Transactions in RCCS}.
\newblock In {\em In Proc. of CONCUR, LNCS 3653}, pages 398--412. Springer,
  2005.

\bibitem{Qing}
Qing Dong.
\newblock A bisimulation approach to verification of molecular implementations
  of formal chemical reaction networks.
\newblock Master's thesis, Stony Brook University, 2012.

\bibitem{douglas2012logic}
Shawn~M. Douglas, Ido Bachelet, and George~M. Church.
\newblock A logic-gated nanorobot for targeted transport of molecular payloads.
\newblock {\em Science}, 335:831--834, 2012.

\bibitem{DOL13}
Jonathan P.~K. Doye, Thomas~E. Ouldridge, Ard~A. Louis, Flavio Romano, Petr
  \v{S}ulc, Christian Matek, Benedict E.~K. Snodin, Lorenzo Rovigatti, John~S.
  Schreck, Ryan~M. Harrison, and William P.~J. Smith.
\newblock Coarse-graining {DNA} for simulations of {DNA} nanotechnology.
\newblock {\em Physical Chemistry Chemical Physics}, 15:20395--20414, 2013.

\bibitem{EN94}
Javier Esparza and Mogens Nielsen.
\newblock {Decidability issues for Petri nets}.
\newblock {\em Petri Nets Newsletter}, 94:5--23, 1994.

\bibitem{BNGL}
James~R. Faeder, Michael~L. Blinov, and William~S. Hlavacek.
\newblock Rule-based modeling of biochemical systems with {BioNetGen}.
\newblock {\em Methods in Molecular Biology}, 500:113--67, 2009.

\bibitem{Grun}
Casey Grun, Karthik Sarma, Brian Wolfe, Seung~Woo Shin, and Erik Winfree.
\newblock A domain-level {DNA} strand displacement reaction enumerator allowing
  arbitrary non-pseudoknotted secondary structures.
\newblock {\em arXiv preprint arXiv:1505.03738}, 2015.
\newblock Originally appeared in Verification of Engineered Molecular Devices
  and Programs (VEMDP) 2014.

\bibitem{gu2010proximity}
Hongzhou Gu, Jie Chao, Shou-Jun Xiao, and Nadrian~C Seeman.
\newblock A proximity-based programmable {DNA} nanoscale assembly line.
\newblock {\em Nature}, 465:202--205, 2010.

\bibitem{Heiner2008}
Monika Heiner, David Gilbert, and Robin Donaldson.
\newblock {Petri} nets for systems and synthetic biology.
\newblock {\em Lecture Notes in Computer Science}, 5016:215--264, 2008.

\bibitem{Holzmann1997}
Gerard~J Holzmann.
\newblock {The model checker SPIN}.
\newblock {\em {IEEE Transactions on Software Engineering}}, 23:279--295, 1997.

\bibitem{PEM99}
Petr Jan{\v{c}}ar, Javier Esparza, and Faron Moller.
\newblock Petri nets and regular processes.
\newblock {\em Journal of Computer and System Sciences}, 59:476--503, 1999.

\bibitem{Johnson2016}
Robert~F. Johnson, Qing Dong, and Erik Winfree.
\newblock Verifying chemical reaction network implementations: A bisimulation
  approach.
\newblock In Yannick Rondelez and Damien Woods, editors, {\em DNA Computing and
  Molecular Programming}, volume 9818 of Lecture Notes in Computer Science
  (LNCS), pages 114--134. Springer, 2016.

\bibitem{Kim2006}
Jongmin Kim, Kristin~S. White, and Erik Winfree.
\newblock Construction of an {\it in vitro} bistable circuit from synthetic
  transcriptional switches.
\newblock {\em Molecular Systems Biology}, 2:68, 2006.

\bibitem{DSD}
Matthew~R. Lakin and Andrew Phillips.
\newblock {\em Visual DSD}.
\newblock Microsoft Research.

\bibitem{Lakin}
Matthew~R. Lakin, Darko Stefanovic, and Andrew Phillips.
\newblock Modular verification of chemical reaction network encodings via
  serializability analysis.
\newblock {\em Theoretical Computer Science}, 632:21--42, 2016.

\bibitem{leroy2009}
Xavier Leroy.
\newblock Formal verification of a realistic compiler.
\newblock {\em Communications of the ACM}, 52:107--115, 2009.

\bibitem{Lipton1976}
Richard Lipton.
\newblock The reachability problem requires exponential space.
\newblock {\em Research Report 62, Department of Computer Science, Yale
  University, New Haven, Connecticut}, 1976.

\bibitem{Mayr1981}
Ernst Mayr.
\newblock Persistence of vector replacement systems is decidable.
\newblock {\em Acta Informatica}, 15:309--318, 1981.

\bibitem{milner1989}
Robin Milner.
\newblock {\em Communication and Concurrency}.
\newblock Prentice Hall, 1989.

\bibitem{Montagne2011}
Kevin Montagne, Raphael Plasson, Yasuyuki Sakai, Teruo Fujii, and Yannick
  Rondelez.
\newblock Programming an {\it in vitro} {DNA} oscillator using a molecular
  networking strategy.
\newblock {\em Molecular Systems Biology}, 7:466, 2011.

\bibitem{QSW11}
Lulu Qian, David Soloveichik, and Erik Winfree.
\newblock Efficient {Turing}-universal computation with {DNA} polymers.
\newblock In Yasubumi Sakakibara and Yongli Mi, editors, {\em {DNA Computing
  and Molecular Programming}}, volume 6518 of Lecture Notes in Computer Science
  (LNCS), pages 123--140, 2011.

\bibitem{Rival04}
Xavier Rival.
\newblock Symbolic transfer function-based approaches to certified compilation.
\newblock In {\em {Proceedings of the 31st ACM SIGPLAN-SIGACT Symposium on
  Principles of Programming Languages}}, POPL '04, pages 1--13, New York, NY,
  USA, 2004. ACM.

\bibitem{Schaeffer}
Joseph Schaeffer.
\newblock {\em Stochastic Simulation of the Kinetics of Multiple Interacting
  Nucleic Acid Strands}.
\newblock PhD thesis, California Institute of Technology, 2013.

\bibitem{SS00}
Philippe Schnoebelen and Natalia Sidorova.
\newblock Bisimulation and the reduction of {Petri} nets.
\newblock {\em Lecture Notes in Computer Science}, 1825:409--423, 2000.

\bibitem{S11}
Seung~Woo Shin.
\newblock Compiling and verifying {DNA}-based chemical reaction network
  implementations.
\newblock Master's thesis, California Institute of Technology, 2011.

\bibitem{SSW10}
David Soloveichik, Georg Seelig, and Erik Winfree.
\newblock {DNA} as a universal substrate for chemical kinetics.
\newblock {\em Proceedings of the National Academy of Sciences},
  107:5393--5398, 2010.

\bibitem{thachuk2012space}
Chris Thachuk and Anne Condon.
\newblock Space and energy efficient computation with {DNA} strand displacement
  systems.
\newblock In Darko Stefanovic and Andrew Turberfield, editors, {\em {DNA
  Computing and Molecular Programming}}, volume 7433 of Lecture Notes in
  Computer Science (LNCS), pages 135--149. Springer, 2012.

\bibitem{Klavins1}
David Thorsley and Eric Klavins.
\newblock Model reduction of stochastic processes using {Wasserstein}
  pseudometrics.
\newblock In {\em American Control Conference}, pages 1374--1381, 2008.

\bibitem{Klavins2}
David Thorsley and Eric Klavins.
\newblock Approximating stochastic biochemical processes with {Wasserstein}
  pseudometrics.
\newblock {\em IET Systems Biology}, 4:193--211, 2010.

\end{thebibliography}
\end{document}